%% file: IEEE_TASE.tex
\documentclass[lettersize,journal]{IEEEtran}
\usepackage[numbers]{natbib}
\usepackage[hidelinks]{hyperref}

\usepackage{amsmath,amsfonts}
\usepackage{algorithmic}
\usepackage{algorithm}
\usepackage{array}
\usepackage[caption=false,font=normalsize,labelfont=sf,textfont=sf]{subfig}
\usepackage{textcomp}
\usepackage{stfloats}
\usepackage{url}
\usepackage{verbatim}
\usepackage{graphicx}

\usepackage{amsmath}
\usepackage{amsthm}
\usepackage{amssymb}

\usepackage{xr}
\makeatletter
\newcommand*{\addFileDependency}[1]{
\typeout{(#1)}
\@addtofilelist{#1}
\IfFileExists{#1}{}{\typeout{No file #1.}}
}\makeatother
\newcommand*{\myexternaldocument}[1]{%
\externaldocument{#1}%
\addFileDependency{#1.tex}%
\addFileDependency{#1.aux}%
}
\myexternaldocument{Appendix}

\newcommand{\probP}{\text{I\kern-0.15em P}}
\newcommand{\problem}{\text{PkBP }}
\newcommand{\problemfull}{\text{Probabilistic k-Bins Packing }}

\newtheorem{lemma}{Lemma}
\newtheorem{theorem}{Theorem}
\newtheorem{definition}{Definition}
\newtheorem{properties}{Properties}

\makeatletter
\newenvironment{procedure}[1][tb]{
    \renewcommand{\ALG@name}{Procedure}
   \begin{algorithm}[#1]
  }{\end{algorithm}}
\makeatother
\allowdisplaybreaks

\begin{document}

\title{Hotspot-Aware Scheduling of Virtual Machines with Overcommitment for Ultimate Utilization in Cloud Datacenters}

\author{Jiaxi Wu, Pavel Popov, Wenquan Yang, Andrei Gudkov, Elizaveta Ponomareva, Xinming Han, Yunzhe Qiu, Jie Song, Stepan Romanov 
\IEEEcompsocitemizethanks{
\IEEEcompsocthanksitem 
Xinming Han, Yunzhe Qiu, and Jie Song are from Peking university. \protect\\
Jiaxi Wu, Pavel Popov, Wenquan Yang, and Stepan Romanov are from Huawei Technologies Company Ltd  \protect\\
E-mail: stepan.romanov@skolkovotech.ru
}}

\makeatletter
\def\ps@IEEEtitlepagestyle{%
  \def\@oddfoot{\mycopyrightnotice}%
  \def\@oddhead{\hbox{}\@IEEEheaderstyle\leftmark\hfil\thepage}\relax
  \def\@evenhead{\@IEEEheaderstyle\thepage\hfil\leftmark\hbox{}}\relax
  \def\@evenfoot{}%
}
\def\mycopyrightnotice{%
  \begin{minipage}{\textwidth}
  \centering \scriptsize
  Copyright~\copyright~2024 IEEE. Personal use of this material is permitted. Permission from IEEE must be obtained for all other uses, in any current or future media, including\\reprinting/republishing this material for advertising or promotional purposes, creating new collective works, for resale or redistribution to servers or lists, or reuse of any copyrighted component of this work in other works by sending a request to pubs-permissions@ieee.org.
  \end{minipage}
}
\makeatother

\markboth{IEEE TRANSACTIONS ON AUTOMATION SCIENCE AND ENGINEERING}{}

\maketitle

\begin{abstract}
We address the problem of under-utilization of resources in datacenters during cloud operations, specifically focusing on the challenge of online virtual machine (VM) scheduling. Rather than following the traditional approach of scheduling VMs based solely on their static flavors, we take into account their dynamic CPU utilization. We employ $\Gamma$-robustness theory to manage the dynamic nature and introduce a novel variant of bin packing - \problemfull (\problem\!\!), which theoretically protects the Physical Machines (PMs) from hotspots formation within a specified probability $\alpha$. We develop a scheduling algroithm named CloseRadiusFit and cold-start AI-based prediction algorithms for the online version of \problem\!\!. To verify the quality of our approach towards the optimal solutions, we solve the Offline \problem problem by designing a novel Mixed Integer Linear Programming (MILP) model and a combination of numerical upper and lower bounds. Our experimental results demonstrate that CloseRadiusFit achieves narrow gaps of 1.6\% and 3.1\% when compared to the lower and upper bounds, respectively.

\end{abstract}

\def\abstractname{Note to Practitioners}
\begin{abstract}

A growing trend in the cloud industry involves overcommitting VMs on PMs. While this approach can ease the problem of low utilization of resources in datacenters, it also introduces a higher risk of hotspots due to resource contention and competition among VMs. In this work, we propose a novel method that leverages $\Gamma$-robustness theory and introduce effective heuristics to achieve ultimate utilization of datacenter resources while ensuring desirable service quality. We validate our approach using real-world production data from Huawei Cloud, improving resource utilization by 125\% over traditional flavor-based allocation methods, while maintaining the occurrence of hotspots below 5\% ($\alpha=0.05$). Our solution only requires VMs' real utilization data that is typically already collected in cloud providers' production environments. Therefore, with minimal modifications to the existing scheduling system, cloud providers can easily implement our solution and reap its benefits. Moreover, in cases of the absence of historical utilization data for VMs (cold-start), we use machine learning to predict VM utilization statistics for our approach.

\end{abstract}


\section{Introduction}\label{sec:introduction}

Cloud service providers offer Virtual Machines (VMs) as online computing product to users over the Internet. The main resource for VM leasing is the central processing unit (CPU). Users purchase a VM with a fixed amount of CPU cores, which is determined by the VM flavor - for instance, 1 core, 2 cores, 4 cores, etc.
The placement of these VMs onto physical machines (PMs) becomes a critical problem for cloud vendors, given the enormous number of VMs created daily.
If the sequence of VMs is known beforehand, assigning VMs to PMs is closely related to the classical bin packing problem, where VMs (items) need to be assigned to PMs (bins) based on their flavor sizes \cite{delorme2016bin, christensen2017approximation}. The problem, known for its NP-hard nature, can be efficiently solved using production systems like Protean \cite{hadary2020protean} or heuristics like BalCon \cite{gudkov2023balcon}. However, taking into account the real utilization of VMs, rather than just their flavor values, can significantly improve resource utilization.

Studies have indicated that the average utilization of cloud data centers worldwide is low, estimated to be between 15\%-20\% \cite{hsieh2020utilization}, due to users utilizing fewer CPU cores than specified in the purchased flavor. Keeping a PM with such low utilization turned on is wasteful in terms of hardware and power consumption, as an idle PM consumes around 70\% of its peak power \cite{beloglazov2012energy}. To enhance data center utilization, providers employ resource "overcommitment" -- a policy that involves allocating VMs to PMs beyond their nominal capacity. However, the overcommitment policy comes at the cost of the risk of resource contention, potentially leading to "hotspots" - situations where multiple VMs simultaneously demand more resources than are available on a PM. Such situations can lead to violation of Service Level Agreements (SLAs) and are strongly undesirable for cloud service providers.

Conventional approaches to overcommitment involve dynamically relocating VMs in the case of hotspots detection or when a host's CPU utilization approaches an empirical threshold \cite{saxena2021op, dabbagh2016energy, beloglazov2010adaptive}. Although these methods can mitigate immediate resource conflicts, they often result in frequent VM migrations, imposing significant overhead on cloud management systems \cite{kumar2019comprehensive}. Therefore, it is crucial to strike a balance between flavor-based allocation, which may result in under-utilization, and dynamic approaches, which can generate numerous hotspots. To improve the effectiveness of dynamic approaches, researchers explored predictive techniques, such as the Resource Central predictor \cite{cortez2017resource}, that utilizes percentiles of VM utilization, and N-sigma predictors \cite{bashir2021take}, which assume Gaussian distribution of utilization. However these methods struggle to provide consistent guarantees and can lead to a high probability of hotspots — up to 18\% when the required per-machine violation rate is 5\% \cite{bashir2021take}. 
Consequently, there is a pressing need for robust methods that achieve optimal resource utilization while maintaining control over hotspot occurrences.

A common approach to providing predefined levels of guarantees to the optimization problems relies on probability theory and stochastic programming with chance-constrained formulations \cite{charnes1963deterministic, nemirovski2007convex, calafiore2006probabilistic}. Even for continuous optimization problems, the inclusion of probabilistic constraints often leads to non-convex feasibility regions and NP-hardness \cite{kuccukyavuz2022chance}. Despite these challenges, stochastic bin-packing with chance constraints was explored both theoretically  \cite{zhang2018ambiguous,song2014chance} and practically, ranging from networking  \cite{kleinberg1997allocating, wang2011consolidating} and cloud computing problems \cite{cohen2019overcommitment, martinovic2021mathematical,yan2022solving} to planning of surgery rooms \cite{denton2010optimal, shylo2013stochastic,zhang2020branch}. Typically, these problems are solved by sample averaging approximation using big-M constraints, which are computationally hard and lack interpretation of bin loads for heuristic applications. Advanced scheduling techniques like reinforcement learning, despite their flexibility, suffer from the same problems of computational complexity, interpretability, and notable optimality gaps \cite{sheng2023learning, guo2024intelligent,xiong2023multi}. 

The most relevant approach to our work is the sub-modular bin-packing problem (SMBP), that transforms probabilistic constraints into deterministic capacity constraints \cite{cohen2019overcommitment}. The authors minimize the number of active hosts in a cluster using BestFit heuristic, introduce theoretical bounds, and provide guarantees on the occurrence of hotspots using sub-modular capacity constraint. These constraints are formulated for VM utilization following either Gaussian or \textit{generic} probability distributions. Building upon this foundation, subsequent research expanded the SMBP framework by including the introduction of additional bounds \cite{martinovic2021mathematical}, the development of branch-and-price solvers \cite{xu2023branch}, and the formulation of problems with modified objective functions \cite{yan2022solving}.

In this work, we propose two major modifications to previous research: 1) we use $\Gamma$-robustness theory to provide guaranties on hotspots, extend family of Gaussian distributions to encompass all symmetrical distributions, and formulate the problem with linear deterministic capacity constraint, unlike the previously used non-linear sub-modular constraints and 2) we introduce a novel \problemfull (\problem\!\!) problem, which accounts for a finite number of hosts in a cluster and uses it as an algorithm parameter. This formulation corresponds to scenarios where we aim to maximize the efficiency of cluster utilization by working with a fixed number of PMs. 

We summarize the main contributions as follows:
\begin{itemize}
    \item  
    We formulate the novel \problem problem and propose an approach towards its solution using the $\Gamma$-robustness theory \cite{BERTSIMAS}. To satisfy the requirements of the theory on the symmetry of a distribution, we propose a procedure for obtaining a stochastic dominant symmetrical random variable of a given empirical distribution. Our approach ensures an equal level of guarantees for all PMs, overcoming limitations of previous work where guarantees were uncontrollable for each PM \cite{BERTSIMAS}.
    \item We study the \problem problem under two scenarios: Hot start and Cold start. This approach significantly enhances the adaptability of our solution to various real-world cloud environments, addressing both situations with and without historical VM utilization data.
    \item To estimate the quality of heuristics, we solve the Offline \problem problem and propose a novel Mixed-Integer Linear Programming Model (MILP) that provides globally optimal solutions for small-sized problems. To limit bounds on optimal solutions for large problem instances, we propose numerical upper and lower bounds, namely PrefixUB and CloseRadiusLB, with a remarkably tight gap of 1.5\% between them.
    \item We develop an online algorithm, CloseRadiusFit, with performance close to that of the lower bound, with a gap of 1.6\%.
    \item We test our approach with real-trace data from Huawei Cloud, demonstrating substantial improvements in PM utilization efficiency with almost hotspot-free allocation. We published both the dataset\footnote{Dataset HW Cloud-2022: \url{https://github.com/huaweicloud/HUAWEICloudPublicDataset}} and the codes for algorithms and bounds\footnote{Code: \url{https://github.com/andreigudkov/CloseRadiusFit}}. 
\end{itemize}

The rest of the paper is structured as follows. Section \ref{sec:ProblemStatement} introduces the problem statement of the proposed \problem problem. In Section \ref{sec:gamma constraint}, we introduce $\Gamma$-robustness theory and propose a procedure to find a symmetric and stochastic dominant random value of a given one. Section \ref{sec:VM placement with trusted predictions} is dedicated to studying the problem under a Hot start assumption, where we propose a MILP model, upper and lower bounds, and a CloseRadiusFit heuristic. In Section \ref{sec:VM placement with untrusted predictions}, we study the problem under the Cold start scenario. In Section \ref{sec:experimental results}, we extensively test our approach with real data from Huawei Cloud and discusses the practical benefits obtained from its application. Finally, Section \ref{sec:conclusion} concludes the paper.

\section{Problem statement}\label{sec:ProblemStatement}
Researchers traditionally address the problem of VM scheduling as a variant of the bin packing problem, where the goal is to pack a finite number of items (VMs) into the fewest number of bins (hosts). However, there are core differences that set our problem of achieving maximal utilization of the datacenter apart. Unlike the classical bin-packing problem, we operate with a fixed and predetermined number of hosts. This modification is dictated by the practical limitations of cloud infrastructure, where the number of available hosts is known and finite. Additionally, rather than relying on a simple summation of the static VM flavor values (item sizes), we consider the dynamic and uncertain nature of VM resource utilization. These distinctions necessitate a novel problem formulation that better reflects the realities of cloud datacenter operations.

We formulate a new online VM scheduling problem, referred to as Probabilistic k-Bins Packing (PkBP). The VMs arrive sequentially, forming a queue $Q=\{VM_1, VM_2, ..., VM_v, ... \}$, and the scheduler places the VMs onto an available set of hosts $H$ indexed by $h \in \{1,2,..,|H|\}$ (see Figure \ref{fig:ProblemStatemente}). For each incoming VM $v$, the scheduler possesses its CPU utilization $u_{v}$ and chooses a destination host using information about previously placed VMs. Importantly, the CPU utilization of future VMs from the queue is unavailable for the scheduler. To handle the uncertainty associated with CPU utilization, we model $u_v$ as a random variable and impose a chance capacity constraint for each host, given by:

\begin{align}
\probP(\sum_{v \in V_h} u_{v} > C) \le \alpha \quad\quad \forall h \in H,
\label{eq:chance_cons}
\end{align}
where $V_h$ is the set of VMs on a host $h$, $C$ is the capacity of the host $h$, and $\alpha$ is the acceptable probability of constraint violation. The goal of the scheduler is to maximize the total number of VMs in the hosts without violating the capacity constraint of Eq.~\ref{eq:chance_cons}. Meanwhile, we forbid the scheduler to skip incoming VMs and rearrange previously placed VMs.

In the offline version of \problem problem, the scheduler accesses information regarding the CPU utilization of all VMs in the queue and searches for the longest "placeable prefix" that allows for a feasible mapping of VMs. To clarify the terminology, we provide the following definitions:
\begin{definition}\label{def:prefix}
A prefix $P_n$ is the set of the first $n$ VMs in the queue $Q$.
\end{definition}
\begin{definition}\label{def:mapping}
A feasible mapping $\mu$ is an assignment of VMs from $P_n$ to hosts such that:
\begin{enumerate}
\item Each VM is assigned to exactly one host.
\item The capacity constraint of each host is satisfied according to Eq. \ref{eq:chance_cons}.
\end{enumerate}
\end{definition}

We consider two scenarios (see Figure \ref{fig:ProblemStatemente}) for \problem problem with different levels of knowledge of CPU utilization of the VMs in the queue : 1) Hot start, when VMs in the queue have historical data on utilization, and 2) Cold start, when we need an oracle to predict utilization statistics for the VMs, since no historical data is available.

To address the stochastic constraint given in Eq. \ref{eq:chance_cons}, we build upon the foundation of $\Gamma$-robustness theory (see Section \ref{sec:gamma constraint}), which is valid for symmetric random variables. To satisfy the symmetry requirement, we introduce a symmetrizer for the empirical distribution in Section \ref{sec:symmetryzer}. In our experiments, we apply the symmetrizer to VM utilization before scheduling.

It is worth noting that combinatorial problems with a finite number of bins were investigated in operational research \cite{sun2022online, labbe2003upper, bruno1985probabilistic}. However, in our problem, we forbid skipping incoming items, unlike in previous works. This condition significantly complicates the problem because we are unable to choose the desired VM sizes from the queue and have to place all the incoming VMs.

\begin{figure}[ht] 

\centering
\centering
\includegraphics[width=\columnwidth]{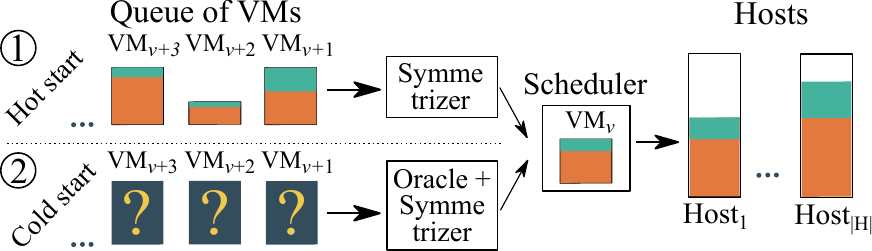}

\caption{Illustration of the problem statement of Online \problem problem with two scenarios: 1) Hot start with known historical utilization of VMs and 2) Cold start which requires oracle to predict VMs statistics for the scheduler. Symmetrizer of VMs distribution is used to satisfy requirements of $\Gamma$-robustness theory.}

\label{fig:ProblemStatemente}
\end{figure}

\section{\texorpdfstring{{$\boldsymbol{\Gamma}$}-robustness constraint}{Γ-Robustness constraint}}\label{sec:gamma constraint}
\subsection{\texorpdfstring{Application of $\Gamma$-robustness theory in VM scheduling}{Application of Γ-robustness theory in VM scheduling}}

To solve the \problem problem, we employ the $\Gamma$-robustness theory by \citet{BERTSIMAS} and formulate deterministic capacity constraint instead of probabilistic one from Eq. \ref{eq:chance_cons}. The theory considers bounded, independent, and symmetrically distributed random variable $u_v$ within the interval  $[uc_v-ur_v,uc_v+ur_v]$, 
where $uc_v$ is the center value and $ur_v$ is the radius value. The theory follows the rationale that multiple independent random variables are unlikely to take the maximum value at the same time. Authors introduce the parameter $\Gamma$, which, in our case, can be interpreted as the maximal number of VMs in a host that are allowed to reach peak utilization simultaneously. Theory binds $\Gamma$ with the probability of constraint violation $\alpha$ from Eq. \ref{eq:chance_cons}. Therefore, $\Gamma$ allows to control trade-off between amount of used space in a host and probability of constraint violation. The more VMs are tolerated to take maximal value, the more space we should allocate for them, and the lower the probability of constraint violation $\alpha$. In the most conservative case, when $\Gamma$ is chosen to be equal to the number of VMs on the host, the allocation policy turns to maximal utilization, reserving space for all VMs at their respective maximum utilization values. This approach ensures a completely safe situation ($\alpha=0$) at the cost of high space reservation. 
To introduce deterministic constraint, we  define a division of VMs on a host $h$ into two sets: $MaxSet_{h}$ and $MinSet_{h}$.
\begin{definition}
$MaxSet_{h}$ is the set comprising the $\Gamma$ VMs with the largest radiuses on a host $h$ and $MinSet_{h}$ is the set of the rest VMs on the host: $V_h  \:\symbol{92} \: MaxSet_{h}$.
\end{definition}

Using this division and $\Gamma$-robustness theory, one can change the probabilistic constraint of Eq. \ref{eq:chance_cons} by a deterministic constraint:
\begin{align}
h^{load}=\overbrace{\sum\limits_{v \in V_h} uc_{v}}^{uc} + \overbrace{\sum\limits_{v \in MaxSet_h} ur_v}^{ur} \le C. 
\label{eq:gamma_constraint}
\end{align}

This transformation is derived from Eq. 4 of the original paper \cite{BERTSIMAS} under the assumption that $\Gamma$ is a positive integer. 
The probabilistic load $h^{load}$ consists of two terms: 1) the sum of the center values of all VMs, denoted as $uc$, and 2) the probabilistic overhead, denoted as $ur$, which is the sum of the largest $\Gamma$ radiuses of VMs on host $h$. 
Bertsimas and Sim demonstrate that if Eq. \ref{eq:gamma_constraint} holds, then probability of violation from constraint of Eq. \ref{eq:chance_cons} can be bounded as:
\begin{align}
\probP\left(\sum_{v \in V_h} u_{v} > C\right) \le B(N,\Gamma),
\label{eq:B_chance_cons}
\end{align}
where $N=|V_h|$ is the number of VMs on the host $h$, and $B(N,\Gamma)$ (see Appendix \ref{appendix: computation of Gamma and widetilde Gamma}, Eq. \ref{eq:B_appendix}) is the function determined by the model. Result of Eq. \ref{eq:B_chance_cons} follows from combination of Proposition 2a and Theorem 3a of the original paper \cite{BERTSIMAS}.

To obtain the value of $\Gamma$ for a required probability $\alpha$, one can numerically solve $B(N,\Gamma)=\alpha$ and acquire $\Gamma(N,\alpha)$ (see Figure \ref{fig:GammaExample}a). 
For a fixed $N$, a smaller $\alpha$ (larger guarantees) results in more VMs taking their maximal values. Assuming $\alpha$ is fixed for all hosts, we denote $\Gamma(N,\alpha) = \Gamma(N)$. This notation implies that whenever the number of VMs ($N$) on a host changes, we need to recalculate both $\Gamma$ and the load $h^{load}$. In the recalculation, the $uc$ component changes with any alteration in $N$, whereas the $ur$ component varies in accordance with changes of $MaxSet_{h}$.

We propose an example to compare three approaches to load calculation of the host based on: 1) flavor sizes, 2) maximal utilization, and  3) stochastic load from Eq. \ref{eq:gamma_constraint}. Consider a host with capacity of 5 cores and 4 VMs of two types of flavors (see Figure \ref{fig:GammaExample}b): 1) $VM_1, VM_2$  with 2 cores each, and 2) $VM_3, VM_4$ with 1 core each. The VMs have different CPU utilization ranges defined by ($uc_v$, $ur_v$): $VM_1$=(1.4, 0.5),  $VM_2$=(0.7, 0.6),  $VM_3$=(0.4, 0.4),  $VM_4$=(0.7, 0.3). The total flavor demand of the VMs is 6 cores, which exceeds the 5-core capacity of the host.  However, by using the maximal CPU utilization of VMs ($MaxUtilization$),  $\sum_{v=1}^{4}{uc_v+ur_v}=5.0$, it is possible to fit all of them into the host. 

The implementation of the stochastic approach allows to further save space and add more VMs. For illustration purposes of $h^{load}$ calculation, we assume that $\Gamma(4)=2$, according to Eq. \ref{eq:gamma_constraint}, we have a probabilistic overhead of two VMs with the largest radiuses $ur=0.5+0.6=1.1$ and sum of centers $uc=1.4+0.7+0.4+0.7=3.2$. Thus, probabilistic constraint in Eq. \ref{eq:gamma_constraint} gives $h^{load}=uc+ur=3.2+1.1=4.3$. This means that using the stochastic approach, we can save 14\% of space compared to maximum utilization and place one more VM such as $VM_2$, $VM_3$ or $VM_4$ if $\Gamma(5)=2$. In this example, $MaxSet_h = \{VM_1, VM_2\}$ and $MinSet_h = \{VM_3, VM_4\}$. 

The Results and discussion section of the paper reveals that the $\Gamma$-robust approach allows for an increase in datacenter utilization by about 125\% compared to the flavor-based approach, and by 18\% compared to the maximal usage approach when $\alpha=0.05$.
 
\begin{figure}[ht] 

\centering
\centering
\includegraphics[width=\columnwidth]{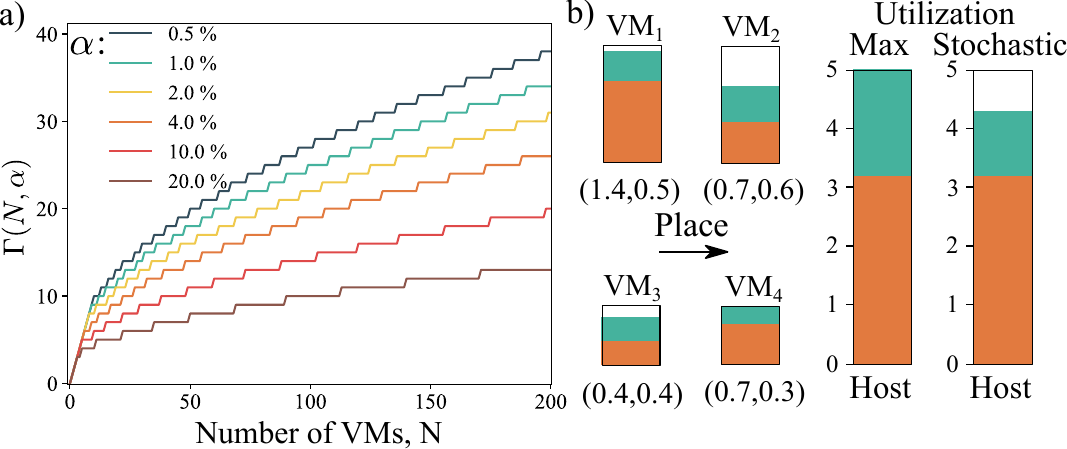}
\caption{ a) Plots of $\Gamma(N,\alpha)$ function, where each curve corresponds to different value of $\alpha$. b) Example of calculation of $\Gamma$-robust constraint, assuming $\Gamma(4)=2$. The stochastic approach allows to decrease required capacity by 14\% compared to maximum utilization. Given 4 VMs with different values of ($uc$, $ur$), maximum utilization requires capacity of 5, whereas probabilistic load $h^{load}$ from Eq. \ref{eq:gamma_constraint} is 4.3.}
\label{fig:GammaExample}
\end{figure}

\subsection{Symmetrizer - a procedure for stochastic dominant symmetrical random variable} 
\label{sec:symmetryzer}
Even though the VM utilization may have a non-symmetric distribution, we can fulfill the symmetry requirement of $\Gamma$-robustness theory by introducing a stochastic-dominant symmetric random variable $\hat{u}$ derived from a given random variable $u$.

\begin{definition}\label{def:stochastic_dominant}
$\hat{u}$ stochastically dominates $u$, denoted as $\hat{u} \succeq u$, if their cumulative distribution functions are related as $F_{\hat{u}}(x) \le F_{u}(x) \quad \forall x$.
\end{definition}

Our objective is to find a symmetrical $\hat{u} \succeq u$ with a minimal center value compared to $u$, such that $\hat{u}c=uc+s_c$, where $s_c$ is the minimal shift required to make $\hat{u}$ symmetrical. It is important to note that due to the construction process, the proposed $\hat{u}$ maintains the same maximum value as $u$ ( $\max(\hat{u}) = \max(u)$).

A symmetrical distribution function $F_{\hat{u}}(x)$ must be equal to its symmetrical reflection $F_{sym(\hat{u})}(x)=1 - F_{\hat{u}}(2\hat{u}c-x)$, and we propose constructing $F_{\hat{u}}(x)$ from $F_u(x)$ and its symmetrical reflection $F_{sym(u)}(x)$.

\begin{figure}[t]
\centering
\includegraphics[width=\columnwidth]{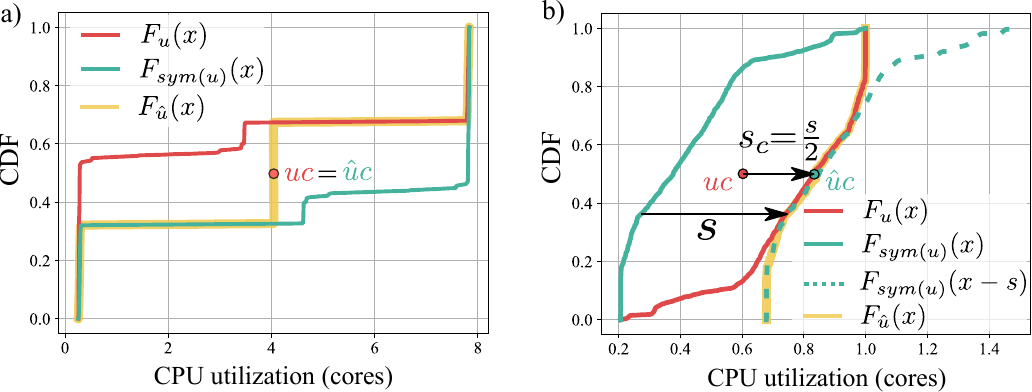}
\caption{Illustration for construction of a symmetrically dominating $F_{\hat{u}}(x)$ (yellow) of given empirical $F_u(x)$ (red) using symmetrical reflection $F_{sym{(u)}}(x)$ (green).  a) Symmetric reflection $F_{sym(u)}(x)$ dominates $F_{u}(x)$ b) Symmetric reflection $F_{sym(u)}(x)$ doesn't dominate  $F_{u}(x)$}
\label{fig:Symmetrizer}
\end{figure}

In the simple case presented in Figure \ref{fig:Symmetrizer}a, $F_{sym(u)}$ (green) dominates $F_u(x)$ (red), leading to an obvious $F_{\hat{u}}(x)$ (yellow), which is a combination of half $F_u(x)$ and half $F_{sym(u)}$. By construction, $F_{\hat{u}}(x)$ is symmetrical and dominates $F_u(x)$ according to Definition \ref{def:stochastic_dominant}. Otherwise, as in Figure \ref{fig:Symmetrizer}b, we shift the symmetrical reflection $F_{sym(u)}(x)$ (green) by $s$ until $F_{sym(u)}(x-s)$ (dashed green) starts to dominate $F_u(x)$, and again comprise the final $F_{\hat{u}}(x)$ (yellow) from the obtained halves of the distributions. In the case of Figure \ref{fig:Symmetrizer}a, $uc=\hat{u}c$, whereas in the case of Figure \ref{fig:Symmetrizer}b, the minimal shift $s_c$ is found as half of the maximal distance between CPU utilization of $F_u(x)$ and its symmetrical reflection $F_{sym(u)}(x)$. We formulate this fact in the following theorem and prove it in Appendix \ref{thm:proof_minimal_s_fordominance}.

\begin{theorem}\label{thm:minimal_s_fordominance}
For a given $u$, the dominating $\hat{u}$ with minimal center has $s_c = \min{{\frac{s}{2}:F_u(x) \ge F_{sym(u)}(x-s)}}$.
\end{theorem}

For the stochastic constraint given by Eq. \ref{eq:gamma_constraint}, we need $\hat{u}c$ and $\hat{u}r$ values, which are obtained in Procedure \ref{alg:Symmetrizer}, where the shift $s_c$ is calculated as half of the largest distance (row 2) between the empirical distribution $u$ and its symmetrical reflection $u_{sym}$ (row 1).

\begin{procedure}[ht!]
    \caption{Symmetrizer}
    \label{alg:Symmetrizer}
    \textbf{Input}:  array with realizations of $u=[u_1, u_2, ..., u_t]$ with center $uc$ and radius $ur$\\
    \textbf{Output}: $\hat{u}c$ and $\hat{u}r$ of dominating symmetrical $\hat{u}$
    
    \begin{algorithmic}[1] 
        \STATE $u_{sym} \leftarrow [ 2\cdot uc - u_i \quad  \text{for} \; i\; \text{in}\; \{1,2,...,t\} ]$
        \STATE $s_c \leftarrow \max(sort(u) - sort(u_{sym}))/2$ 
        \STATE $\hat{u}c \leftarrow uc + s_c$
        \STATE $\hat{u}r \leftarrow ur - s_c$
        \RETURN $\hat{u}c, \hat{u}r$
    \end{algorithmic}
\end{procedure}

\section{VM placement in Hot start scenario } \label{sec:VM placement with trusted predictions}
The roadmap for developing the scheduling algorithm involves four stages. In the first stage, we propose a novel MILP model to obtain the optimal solution for small-scale offline \problem problem. Moving on to the second and third stages, which deal with larger problem instances, we design an upper bound named PrefixUB and a lower bound named CloseRadiusLB to narrow the potential range of the optimal solution in offline problem settings. In the fourth stage, we shift our focus to the online version of \problem problem, and propose heuristics towards its solution.

\subsection{Optimal solution of offline \problem problem with MILP model} \label{subsec: optimal solution}

\begin{table}[t]
\begin{center}
\caption{Notation Summary of MILP}
\label{tab: Model parameters}
\resizebox{\linewidth}{!}{
\begin{tabular}{p{1.4cm} p{7cm}}
\hline
Parameters         & Descriptions    \\ 
\hline
$H$      & The set of indexes for hosts $h \in \{1,2,...,|H|\}$\\
$V$                        & The set of indexes for VMs in queue Q, that is large enough to be implaceable into hosts, $v \in \{1,2,...,|Q|\}.$   \\ 
$uc_v$, $ur_v$ & The values of center and radius of $v$-th VM respectively\\

$C$                     & The total available capacity of the host.                     \\

$ur_{max}$                  & The maximal value of radiuses of all VMs, where $ur_{max}=\max\limits_{v\in V} ur_i$.  \\
$N_{max}$                      & An upper bound on a number of VMs on a host\\ 
$\Gamma_{k}$             & The number of VMs in $MaxSet$ for host with $k$ VMs, where $\Gamma_k=\Gamma(k)$ for $k \in \{1,2,..., N_{max}\}$ \\
\hline
Variables       &    Descriptions   \\ 
\hline
$R_{h,k}$                  & The binary variable indicating whether the number of VMs allocated on the host $h$-th host is $k$.     \\
$y_{v,h}$                  &The binary variable indicating whether the $v$-th VM is allocated on the $h$-th host and belongs to the $MaxSet_{h}$.      \\
$x_{v,h}$                  &The binary variable indicating whether the $v$-th VM is allocated on the $h$-th host and belongs to the $MinSet_{h}$.    \\
$S_h$                      & A variable separating $MinSet_{h}$ and $MaxSet_{h}$: $\max\limits_{v \in MinSet_{h}} ur_v \leq S_h \leq \min\limits_{v \in MaxSet_{h}} ur_v$.    \\ 
\hline
\end{tabular}}
\end{center}
\end{table}

We mathematically define the problem from Section \ref{sec:ProblemStatement} through a  MILP formulation. In our problem statement, we fix the probability of constraint violation for each host to $\alpha$, based on the cloud SLA requirement. Our approach differs from the work of \citet{BERTSIMAS} where the authors fix $\Gamma$, which can result in different constraint violation probabilities across the hosts in an optimal solution. $\Gamma$-robustness constraint given by Eq. (\ref{eq:gamma_constraint}) states that among $N$ VMs, only $\Gamma(N)$ belong to $MaxSet_h$ and contribute their radius and center values to the host load, whereas the rest of the VMs from $MinSet_h$ are taken only by center values. Therefore, the main challenge addressed by our MILP model is dividing the VMs on a host between $MaxSet_h$ and $MinSet_h$ in order to fix the probability of violation for every host to $\alpha$. We were able to achieve this division in linear fashion.  We summarize the notations for the MILP model in Table \ref{tab: Model parameters}.

The objective function is to maximize the total number of VMs that can be provisioned on the given hosts. We formulate the MILP model for the \problem problem as follows:

\begin{subequations}\label{eq:2}
\begin{align}
\max  \hspace{0.5cm}&\sum_{h\in H}\sum _{v \in V}(x_{v,h}+y_{v,h}), \label{obj}\\
 s.t. \hspace{0.5cm} &\sum_{h\in H}(x_{v,h}+y_{v,h})\leq 1  \hspace{0.5cm}\forall v\in V , \label{cons:c-1}\\
 &\sum_{k=0}^{N_{max}} R_{h,k}=1, \hspace{0.5cm}\forall h \in H, \label{cons:c-2}\\
 &\sum_{v\in V} (x_{v,h}+  y_{v,h})=\sum_{k=0}^{N_{max}} R_{h,k}\cdot k \hspace{0.5cm}\forall h\in H ,\label{cons:c-3}\\
 &\sum_{v\in V} y_{v,h}=\sum_{k=0}^{N_{max}} (R_{h,k}\cdot \Gamma_{k}) \hspace{0.5cm} \forall h\in H,  \label{cons:c-4}\\ 
 & \sum_{h\in H}(x_{v,h}+y_{v,h})\geq \sum_{h\in H}(x_{v+1,h}+y_{v+1,h}) \hspace{0.5cm} \nonumber\\  & \hspace{4cm} \forall v\in V, \label{cons:c-5}\\
 & S_h\geq x_{v,h}\cdot ur_{v} \hspace{0.5cm} \forall v\in V, h\in H \label{cons:c-6}\\
 &  ur_{max}-S_h\geq y_{v,h}\cdot ( ur_{max}-ur_v) \hspace{0.5cm}\forall v\in V, h\in H \label{cons:c-7} \\
 &\sum_{v \in V} (x_{v,h}+ y_{v,h})\cdot uc_{v} +\sum_{v\in V} y_{v,h}\cdot ur_{v}\leq C,  \forall h\in H \label{cons:c-8}\\
 &R_{h,k} \in \{0,1\} \hspace{0.5cm}\forall k\in [0,1,\cdots,N_{max}], h \in H  \label{cons:c-9}\\
 &y_{v,h} \in \{0,1\} \hspace{0.5cm}\forall v\in V, h\in H \label{cons:c-10}\\
 &x_{v,h} \in \{0,1\} \hspace{0.5cm}\forall v\in V, h\in H  \label{cons:c-11} 
\end{align}
\end{subequations}

Constraint (\ref{cons:c-1}) indicates that each VM is either in the $MaxSet_{h}$, or the $MinSet_{h}$, or rejected. Constraint (\ref{cons:c-2}) indicates that the number of VMs placed on each host is uniquely determined. Constraint (\ref{cons:c-3}) ensures that the number of VMs allocated on host $h$ equals $k$. Constraint (\ref{cons:c-4}) guarantees that the number of VMs in the $MaxSet_h$ is exactly $\Gamma_{k}$. Constraint (\ref{cons:c-5}) preserves the order of VMs in queue for allocation and assures that VMs are not skipped. Constraints (\ref{cons:c-6}) and (\ref{cons:c-7}) separate the $MaxSet_{h}$ and the $MinSet_{h}$ to ensure that the minimum radius in the $MaxSet_{h}$ is weakly greater than the maximum in the $MinSet_{h}$. Constraint (\ref{cons:c-8}) guarantees that the capacity constraint remains intact. The remaining constraints define the domains of the variables.

\subsection{Upper bound-PrefixUB} \label{subsec:upper bound2}
Commercial solvers can provide optimal solutions of our MILP model (section \ref{subsec: optimal solution}) for a few hosts in a reasonable time. However, the solvers encounter problems with tens of hosts because of NP-completeness of \problem problem. To estimate the optimal solution for many hosts, we propose a numerical upper bound named PrefixUB. To provide an upper bound on the number of placed VMs, PrefixUB searches for the shortest prefix from the queue that is implaceable in any mapping of VMs to hosts. This prefix has a 
 larger CPU demand than capacity of all hosts in datacenter (see Algorithm \ref{alg:PrefixUB}). By establishing a lower bound on the probabilistic utilization overhead $u_{lb}$, we can determine if the prefix is implaceable. If the sum of the VMs' center value $uc$ and the lower bound on the probabilistic utilization overhead $ur_{lb}$ exceeds the host available capacity, i.e., $uc+ur_{lb} \geq |H|\cdot C$, then the prefix is considered implaceable (row 12).

To obtain a tight estimation on the probabilistic overhead of all hosts $ur_{lb}$, we strive to identify the largest radiuses of VMs from a given prefix $P_i$ that either are placed in $MaxSet$ in any mapping or provide accurate estimation of radiuses of VMs from the $MaxSet$. To identify these VMs, we utilize a lower bound on the cumulative $\Gamma$ of the prefix $\Gamma^{lb} \le \sum_{h=1}^{|H|} \Gamma(|V_h|)$ (row 4). This lower bound is independent of the placement of VMs on hosts and provides a lower limit on the total number of VMs in $MaxSet$. If $\Gamma^{lb}$ increases after adding a new VM to the prefix (row 8), then at least this VM belongs to $MaxSet$, and we increase $ur_{lb}$ by its radius (row 9). The proposed algorithm provides an upper bound on the number of placed VMs for any mapping. We formulate and prove this fact in Theorem \ref{thm:overhead2}, which depends on two assumptions: 1) There exists a concave approximation of $\Gamma(N)$ denoted as $\tilde{\Gamma}(N)$ (see Section \ref{sec:concave}) which is required for calculation of $\Gamma^{lb}$, and 2) $\Gamma^{lb}$ is a lower bound on the number of VMs in $MaxSet$ in any mapping (see Section \ref{subsec:GammaLB}).

\begin{theorem}\label{thm:overhead2}
\textnormal{The computed value $ur_{lb}$ is a lower bound on the probabilistic overhead capacity for any feasible mapping of $Q$. (The proof is in Appendix \ref{appendix: proof of overhead})}
\end{theorem}

\begin{algorithm}[tb] 
    \caption{PrefixUB}
    \label{alg:PrefixUB}
    \textbf{Input}: $Q$ is a queue of VMs  \\
    \textbf{Output}: $ub$ is an upper bound on the number of placed VMs in optimal solution
    \begin{algorithmic}[1] 
        \FOR{$i\leftarrow 1$ to $|Q|$}
            \STATE $P_i \leftarrow$ $i$-th prefix of $Q$ 
            \STATE sort $P_i$ by radiuses in decreasing order
            \STATE $\Gamma^{lb}_j\leftarrow$ GammaLB$(j$-th prefix of $P_i)$ \textbf{for} $j\leftarrow 0$ to $i$ 
            \STATE $uc \leftarrow \sum_{j\leq i} uc_j$
            \STATE $ur_{lb} \leftarrow 0 $ 

            \FOR{$ j\leftarrow 1$ to $i$}
                \IF{$\Gamma^{lb}_j > \Gamma^{lb}_{j-1}$}
                    \STATE $ur_{lb}\leftarrow ur_{lb} + ur_{j}$
                \ENDIF
            \ENDFOR
            \IF{$uc + ur_{lb} > C|H|$} 
                \RETURN $i - 1$
            \ENDIF
        \ENDFOR
        \RETURN $|Q|$
    \end{algorithmic}
\end{algorithm}

\begin{procedure}[ht]
    \caption{GammaLB}
    \label{alg:GammaLB}
    \textbf{Input}: $P$ is a prefix of VMs\\
    \hphantom{\textbf{Input}:} $|H|$ is a number of hosts 
    
    \textbf{Output}: $lb$ is a lower bound on the number of VMs in MaxSet of all hosts in any mapping;
    \begin{algorithmic}[1] 
        \STATE sort $P$ by centers in increasing order
        \STATE $k, n, c\leftarrow 0$
        \STATE $d \leftarrow$ a zero array of length $|H|$ 
        \FOR{$i\leftarrow 1$ to $P$} 
            \STATE $n\leftarrow n + 1$
            \STATE $c\leftarrow c + uc_{i}$ 
            \IF{$c \geq C$} 
                \STATE $c\leftarrow c - C$
                \STATE $d_k\leftarrow n$, $n\leftarrow 0$ 
                \STATE $k\leftarrow k + 1$
            \ENDIF
        \ENDFOR
        \STATE $d_k \leftarrow n$ 
        \RETURN $\sum\nolimits_{k=1}^{|H|} \tilde{\Gamma}(d_k)$
    \end{algorithmic}
\end{procedure}

\subsubsection{Concave approximation of 
\texorpdfstring{$\mathbf{\Gamma}$}{Lg}}\label{sec:concave}
As one can see in Figure \ref{fig:GammaExample}a, the function  $\Gamma(N)$ is not concave. We use the following Linear Programming (LP) model to derive a lower concave approximation $\widetilde{\Gamma}(N)$ of $\Gamma(N)$ for all $N=0,\cdots,N_{max}$:
\begin{align*}
  & \max_{\widetilde{\Gamma}(N)}  \ \sum_{N=0}^{N_{max}}\widetilde{\Gamma}(N) \\
  & 0\le \widetilde{\Gamma}(N) \le \Gamma(N) & 0 \le N \le N_{max}, \\
  & \widetilde{\Gamma}(N+1) - \widetilde{\Gamma}(N) \le \widetilde{\Gamma}(N) - \widetilde{\Gamma}(N-1) & 1 < N < N_{max}.
\end{align*}

More detailed descriptions on computation of \texorpdfstring{$\Gamma$}{Lg} are in Appendix \ref{appendix: computation of Gamma and widetilde Gamma}.

\subsubsection{Lower bound on $|MaxSet|$ of any mapping}
\label{subsec:GammaLB}
To find the lower bound $\Gamma^{lb}(P_n)$ on the number of VMs in $MaxSet$ of any mapping, we use a concave approximation of $\Gamma(N)$ and provide a procedure to construct a packing that \textit{majorizes} all possible placements of VMs to hosts.

\begin{definition}\label{def:density2}
\textnormal{For two non-increasing sequences of non-negative integer numbers:}

\[
\begin{array}{l}
  a_1 \geq a_2 \geq \cdots \geq a_n \\
  b_1 \geq b_2 \geq \cdots \geq b_n
\end{array},\;\;\;\sum\nolimits_{i=1}^n{a_i} = \sum\nolimits_{i=1}^n{b_i},
\]
\textnormal{Sequence $\{a_i\}$ \textit{majorizes} $\{b_i\}$ if the following relation holds:}
    \begin{align*}
        \forall k \leq n: \sum\nolimits_{i=1}^k a_i \geq \sum\nolimits_{i=1}^k b_i\label{eq:density}. 
    \end{align*}
\end{definition}
In the packing problem, $a_i$ and $b_i$ stand for the number of VMs on $i$-th host in two different mappings. Note that such placement strives to place the most VMs in first hosts and even if $a_2 < b_2$, the condition on cumulative sum requires $a_1 + a_2 \ge b_1 + b_2$. 

To produce a packing $\{d_j\}$ that majorizes all possible placements (see Procedure \ref{alg:GammaLB}), we neglect the radiuses of the VMs, sort them by centers in ascending order, and pack the VMs using the Next-fit heuristic with relaxation - if a VM unfits into a host, we split it between that host and the next one (rows 7-10). Such packing majorizes any other placement by construction, since for a given prefix of VMs the packing maximizes $\sum\nolimits_{i=1}^k d_i$ for any $k$.

We use Hardy–Littlewood–Pólya theorem \cite{hardy1952inequalities} that states if $\{a_i\}$ majorizes $\{b_ i\}$ and function $f$ is concave then 
\begin{equation}
    \sum\nolimits_{i=1}^{N} f(a_i) \leq \sum\nolimits_{i=1}^{N} f(b_i)
\end{equation}

Taking $\widetilde{\Gamma}(N)$ as the lower concave approximation of $\Gamma(N)$ (see Section \ref{sec:concave}) we obtain the lower bound

\begin{equation}
    \Gamma^{lb}(P_n)= \left\lceil \sum\limits_{i=1}^{|H|} \widetilde{\Gamma}(d_i) \right\rceil \leq \sum\limits_{i=1}^{|H|} \widetilde{\Gamma}(a_i)  \quad \forall \{a_i\}.
\end{equation}

\subsection{Lower bound-CloseRadiusLB} \label{subsec:lower bound}

To propose a Lower Bound (LB) on the optimal number of placeable VMs from the queue, we consider an offline version of the packing problem where we have full knowledge of the VMs in the queue. We exploit the idea that a well-packed host has VMs with close radiuses. The concept is supported by the following two lemmas.

\begin{lemma} \label{lemma:radiuses seperation}
\textnormal{Suppose we want to divide $2N$ VMs with a constant $uc_v$ into two groups, each consisting of $N$ VMs, while minimizing the probabilistic overhead $ur$. In the optimal solution, one group consists of $N$ VMs with the largest radiuses, and the second group has $N$ VMs with the smallest radiuses. (The proof is in Appendix \ref{appendix: proof of radiuses seperation})}
\end{lemma}

Lemma \ref{lemma:radiuses seperation} supports our claim that VMs with similar radiuses should be packed together.

We further illustrate the advantage of the close radius approach by examining the sum of radiuses of VMs from $MinSet$: $ur_{MinSet} = \sum\nolimits_{v=\Gamma}^{N} ur_v$. The sum  represents an improvement in terms of CPU utilization of $\Gamma$-robustness theory over $MaxUtilization$ approach, where we take any VM $v$ by its maximal value $uc_v+ur_v$. 

\begin{lemma} \label{lemma:best radiuses values}
Among all possible fillings of a host of capacity $C$ with $N$ VMs of fixed $uc = \sum\nolimits_{v=1}^{N}{uc_v}$,
the filling that maximizes $ur_{MinSet}$ is achieved 
when radiuses of all VMs are the same and equal to $ur_{opt}=\left(C - \sum{uc_v}\right)/\Gamma_N$. (The proof is in Appendix \ref{appendix: proof of best radiuses values})
\end{lemma}

Although the constant requirement on the center value is not fully representative, the lemmas reflect the intuition that to reduce the probabilistic overhead  $ur$ (equally to increase the number of placeable VMs), one should strive to place VMs of close radiuses together on the same host. Therefore, sorting VMs by radiuses and packing them into hosts with FirtstFit heuristic is a perspective approach. Based on this idea, to find a lower bound on the optimal number of VMs in the longest placeable prefix, we propose CloseRadiusLB (see Algorithm \ref{alg:CloseRadiusLB}).

In general, CloseRadiusLB finds an upper bound on the prefix length and uses binary search to define the maximal placeable prefix. In the procedures of the algorithm, rows 3-5 find a trivial upper bound on the prefix length \textit{high} by assuming zero radiuses of VMs. To find the length of the maximal placeable prefix $P_n$, we perform a binary search between $low=0$ and the obtained upper bound $high$ (rows 8-17). At each step, to check if the prefix is placeable, we sort VMs by radiuses in descending order and place them one by one with the FirstFit heuristic (rows 11-12). Finally, we have the lower bound on the optimal number of placeable VMs $lb$.

\begin{algorithm}[tb] 
    \caption{CloseRadiusLB}
    \label{alg:CloseRadiusLB}
    \begin{flushleft}
        \textbf{Input}: $Q$ is an infinitely long queue of VMs \\
        \hphantom{\textbf{Input}:} $H$: is a set of hosts \\
        \textbf{Output}: $lb$ is the lower bound on number of placed VMs in optimal solution
    \end{flushleft}
    \begin{algorithmic}[1] 
        \STATE $\textit{low}, \textit{high}, uc_{sum} \gets 0, 0, 0$
        \STATE $\textit{capacity} \gets |H|C$
        \WHILE{$uc_{sum} \le \textit{capacity}$}
        \STATE $uc_{sum} \gets uc_{sum} + uc_{high}$
        \STATE $high \gets high + 1$
        \ENDWHILE
        \STATE $lb \gets 0$
        \WHILE{$\textit{low} < high$}
        \STATE $n \gets (low + high)/ 2$
        \STATE $P_n \gets$ first $n$ VMs from $Q$
        \STATE Sort $P_n$ by radiuses in decreasing order
        \IF{$\textbf{FirstFit}(P_n, H)$}
        \STATE $lb \gets n$
        \STATE $low \gets n + 1$
        \ELSE
        \STATE $high \gets n - 1$
        \ENDIF
        \ENDWHILE
        \RETURN $lb$
    \end{algorithmic}
\end{algorithm}

\subsection{CloseRadiusFit heuristic for online \problem problem } \label{subsec:heuristic}
\begin{algorithm}
    \caption{CloseRadiusFit}
    \label{alg:CloseRadiysFit}
    \begin{flushleft}
        \textbf{Input}: $P_{i-1}$ is the set of already placed VMs \\
        \hphantom{\textbf{Input}:} $i$ is an index of the VM to place \\
        \hphantom{\textbf{Input}:} $|H|$ is a number of hosts \\
        \textbf{Output}: host where to place $i$-th VM 
    \end{flushleft}
    \begin{algorithmic}[1] 
        \STATE \textbf{Phase-1}: predict desired radius range for each host.
        \STATE Sort VMs in $P_{i-1}$ by $ur_i$ in descending order
        \STATE $cap \leftarrow \frac{1}{|H|}\sum_{i\in P_{i-1}}{uc_i}$ \label{alg:CloseRadiysFit:line_hac}
        \STATE $h_{rb}[\:] \leftarrow \text{a zero array of length } |H|$
        \STATE $p \leftarrow 0 $
        \FOR{$h \leftarrow 1$ to $|H|$}
        \STATE $req \leftarrow 0 $
        \WHILE{$req < cap$ and $ p < i$}
            \STATE $req \leftarrow req + uc_p$
            \STATE $p \leftarrow p + 1$
        \ENDWHILE
        \STATE $h_{rb}[h] \leftarrow ur_{p}$
        \ENDFOR
        \STATE \textbf{Phase-2}: use $h_{rb}$ values to find host for VM $i$.
        \STATE Find optimal host $p$: $h_{rb}[p-1] \le ur_i \le h_{rb}[p]$ 
        \FOR{$j \leftarrow p$ down to $1$ }
          \IF{if $i$-th VM fits into host $j$} \RETURN $j$ \ENDIF
        \ENDFOR
        \FOR{$j \leftarrow p$ to $|H|$ }
          \IF{if $i$-th VM  fits into host $j$} \RETURN $j$ \ENDIF
        \ENDFOR
        \RETURN "No available host"
    \end{algorithmic}
\end{algorithm}

In a real-world scenario, the sequence of upcoming VMs is unknown in advance.  Our objective is to design an online algorithm that effectively leverages the key observations from the offline CloseRadiusLB. To accomplish this, we introduce CloseRadiusFit that "imitates" CloseRadiusLB and strives to place VMs with close radius together. Using information of previously placed VMs, CloseRadiusFit predicts the preferable distribution of VMs to hosts by radiuses. At each iteration, to place an $i$-th VM in the host with the closest radius, CloseRadiusFit uses two phases (see Algorithm \ref{alg:CloseRadiysFit}). 

\textit{Phase-1}: We predict the radius bound, $h_{rb}$, for each host using information of previously placed VMs $P_{i-1}$ (rows 2-12). To perform the prediction of the bounds for $|H|$ hosts,  we simulate the placement of the VMs from prefix $P_{i-1}$ into hosts with reduced capacity. This reduction is necessary for accurate distribution estimation of VM radiuses over hosts.
At first, we calculate the reduced host capacity $cap$ that is an average capacity of cumulative centers from prefix $P_{i-1}$ (row \ref{alg:CloseRadiysFit:line_hac}). Then, for each $j$-th host, we estimate its radius bound $h_{rb}[j]$. To do this, we sort VMs from $P_{i-1}$ in descending order by radius and sequentially place them into a host with the new average capacity $cap$ until the host is full (rows 6-11). The radius of the last unplaced VM, $ur_p$, is taken as the upper bound for the $j$-th host (row 12). Based on the prediction, $j$-th host is expected to accommodate VMs $v$ with radiuses $ur_v \in [h_{rb}[p-1], h_{rb}[p]$.

\textit{Phase-2}: Place the newly arrived VM, denoted as VM $i$, into the host with the corresponding radius range. Specifically, we place the new VM in the host with the closest radius based on the computed bounds, $h_{rb}$ (row 15). If the VM does not fit into the host, we search across the other hosts to find one that is able to accommodate it (rows 16-25).

\section{VM placement in Cold start scenario} \label{sec:VM placement with untrusted predictions}
For the Cold start scenario, the utilization history of VMs in the queue is unavailable. To resolve this issue, we employ Gradient Boosting Decision Tree (GBDT) to predict interval values of VMs' utilized capacity based on about 20 available features of a new VM, such as flavor type, region, creation time, etc. The GBDT was chosen as a well-interpretable machine learning approach with successful applications to cloud problems \citep{cortez2017resource}. To fulfill the information gap between the cold start scenario to the hot start scenario, we design two strategies : 1) we place VMs based on the predicted range, and 2) we adjust the predicted range using utilization from VMs' recent activity. The second strategy mixes cold and hot starts, as we initially use the predicted range, and then adjust it with the recently measured utilization.

For the first strategy, to coincide with the format of $\Gamma$-robustness theory, we predict the utilization ranges of VMs. We employ GBDT to classify a VM into one of four groups having utilization ranges: 10-25\%, 10-50\%, 10-75\%, and 10-100\%. To avoid hotspots, we set an extra penalty for the underestimation of utilization in a loss function. We define $r_i = f(\boldsymbol{x_i})-y_i$ as the gap between the predicted and real values, where $f(\cdot)$ is the prediction function, the vector $\boldsymbol{x_i}$ is the input features of VM $i$, and $y_i$ is the real utilization range of VM $i$. To distinguish underestimated values from overestimated, we introduce indicating variables for sign of $r$:  $I_{r<0}$ and $I_{r\geq 0}$. Finally, we define the loss function for prediction as $L=\sum_{v\in V}(w\cdot I_{r<0}\cdot r_v^2+ I_{r\geq 0}\cdot r_v^2)$, where  $w$ is a penalty for underestimation.

As the prediction algorithm from the first strategy is rather conservative and tends to overestimate the VMs' utilization, we need to correct the prediction value for the VMs already running on the hosts. For this purpose, we developed a second strategy named Semi-cold start, where we employed an auto-correction scheme that modifies the predicted interval from the first strategy using the last $t_a$ minutes of historical data. $t_a$ is a hyper-parameter that should be tuned. If $t_a$ is small, the host is exposed to a high risk of hotspots if many VMs arrive in a short period of time. In contrast, if $t_a$ is large, the correction is ineffective. The estimation of the utilization range from historical data converges to actual future utilization according to \citet{cortez2017resource}. With accurate predictions of the utilized capacity interval of VMs, we better estimate the utilization of each host for $\Gamma$-robustness theory and improve the overall packing quality.

\section{Experimental results} \label{sec:experimental results}

\subsection{Datasets and experiment flow}

\begin{figure}
\centering
\includegraphics[width=\columnwidth]{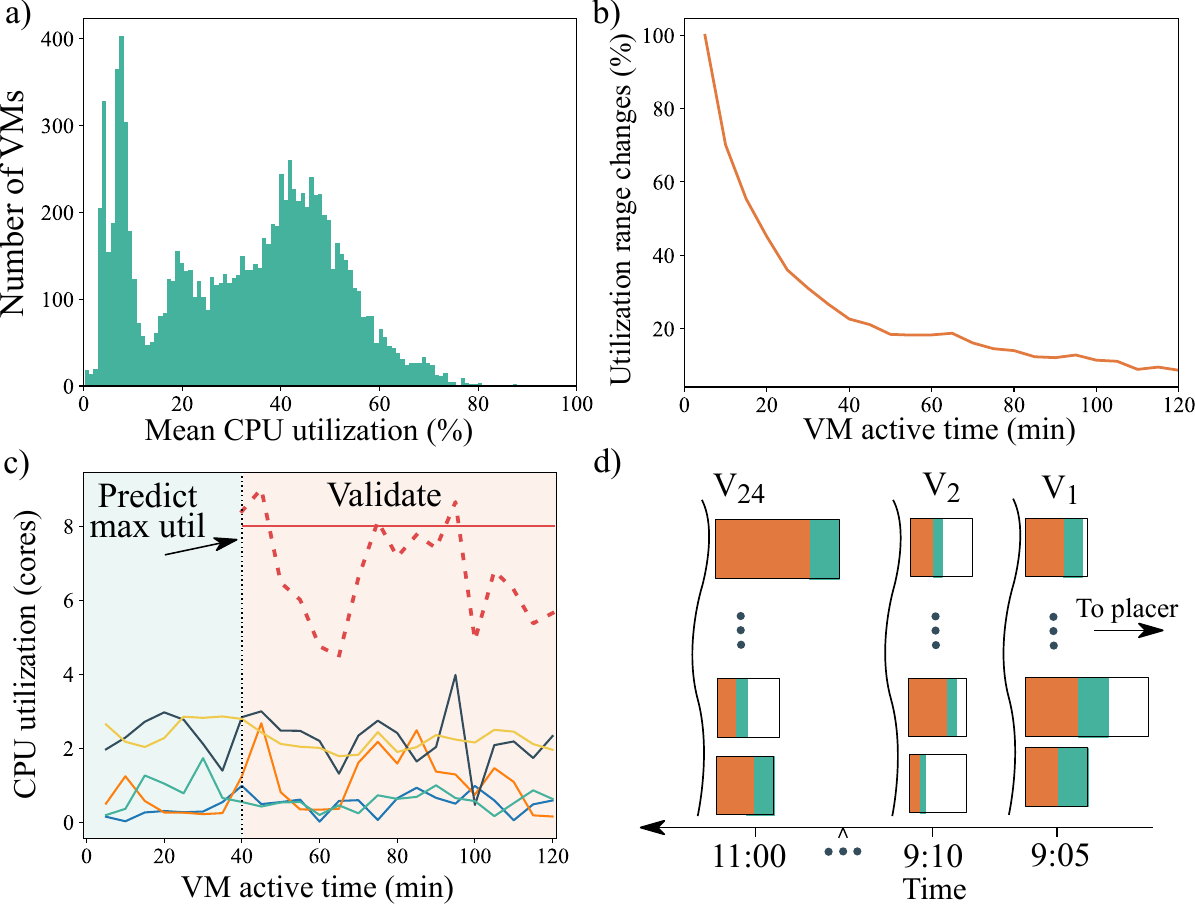}
\caption{a) Distribution of mean utilizations of VMs from the Huawei dataset. b) Illustrates which part of VMs from the dataset changes their utilization range over activation time. c) Demonstrates an example of host load and hotspots calculation using prediction range in green and validation range in red. d) Illustrates trace with arrival of VMs for Cold start scenario placement.}
\label{fig:Dataset_expflow}
\end{figure}

We evaluated our methods for solving the \problem problem using about 10,000 VMs launched from 9:00 AM to 11:00 AM in March 2022 in Huawei Cloud. The dataset consists of flavors with 1, 2, or 4 cores and comprises primarily short-lived VMs with an active time usually less than 24 hours. These VMs were specifically chosen for stress-testing purposes due to their high CPU utilization with a mean of 32\% (the distribution of means is given in Figure \ref{fig:Dataset_expflow}a), which is higher than the reported 15\% in the literature. The reason for this high utilization is attributed to the short lifetime of the VMs, whereas long-lived VMs typically exhibit lower and more steady utilization. Furthermore, the utilization range of the VMs changes over time in terms of minimum and maximum values. However, the majority of variation occurs within the first 40 minutes after the VM launch (as depicted in Figure \ref{fig:Dataset_expflow}b). Specifically, after 40  minutes, only 20\% of the VMs experience changes of utilization range, in contrast to 70\% after 10 minutes.

For the Hot start scenario, we used the initial $t_p$ timestamps of utilization data subsequent to VM launch to make predictions of VMs utilization range, whereas employing the remaining time series of the utilization to validate the prediction. Violation of Eq. \ref{eq:gamma_constraint} during the validation phase indicates the occurrence of a hotspot. Figure \ref{fig:Dataset_expflow}c provides an example for utilization of 5 VMs with $t_p=40$ minutes, where the host capacity is represented by the solid red line, while the real summed utilization is depicted by the dashed-red line. 

For the Cold start scenario, we provide a dataset with traces of VMs referred to as \textbf{HW Cloud-2022}. Each trace is constructed for a fixed number of hosts ranging from 2 to 100 and contains 24 groups $\{V_{i}\}_{i=1}^{24}$ of VMs. Each group of VMs arrives for placement within an interval of 5 minutes (see Figure \ref{fig:Dataset_expflow}d). Every 5 minutes, we measured the CPU utilization for all VMs in the cluster and introduced a new group of arriving VMs. The host capacity in experiments was fixed to $C=44$ cores.

\subsection{Evaluation of performance in Hot start scenario}

Our objective is to maximize the number of successfully placed VMs from the queue onto the hosts. To evaluate the performance of our algorithms and bounds, we examined the relationship between the number of placed VMs and the number of hosts. We conducted 100 experiments for each host number, as shown in Figure \ref{fig:HotStart_results}a. In these experiments, we set $t_p$ to 40  minutes and $\alpha$ to 0.05. The results indicate an almost linear dependence of all algorithms, where the inclination angle represents the number of placed VMs per host $V_{ph}$. It is worth noting that some deviation from perfect linearity persists in the region of 1-20 hosts. The values for $V_{ph}$ of algorithms and Gaps to bounds are presented in Table \ref{tab:Evaluation of the upper bound and lower bound}, where

\begin{equation}
    Gap_{X\to Y}:=\frac{OBJ(Y) - OBJ(X)}{OBJ(Y)} \cdot 100\%.
\label{eq:gap}
\end{equation}

Here, $OBJ(\cdot)$ denotes the number of placed VMs by algorithm $X$ or $Y$. The metric quantifies the performance difference between algorithms and bounds as a percentage gap.

Our PrefixUB and CloseRadiusLB algorithms provide narrow bounds on the optimal solution, with a $V_{ph}$ value of $38.7$ for PrefixUB and $38.1$ for CloseRadiusLB. The small gap of 1.5\% between the bounds indicates their high quality. Therefore, the range for the optimal solution in terms of $V_{ph}$ is $[38.1, 38.7]$ VMs per host.

We also performed a direct search for the optimal solution using our MILP model and the Gurobi solver \cite{gurobi}. The problem is solvable for small instances with about 5 hosts and demonstrates a gap of less than 0.1\% to the CloseRadiusLB. This further validates the high quality of the LB and indicates that the optimal solution should be in close proximity to the LB.

Placement of VMs with similar radiuses together allows to save more space than in other approaches because of the structure of the $\Gamma$-robust constraint (see Sections \ref{sec:gamma constraint} and \ref{subsec:lower bound}). That is why CloseRadiusFit places about 7\% more VMs than well-known heuristics such as FirstFit and RandomFit. Furthermore, the performance of CloseRadiusFit is close to that of LB and UB with gaps of 1.6\% and 3.1\%, respectively. Note that all the heuristics use $\Gamma$-robust constraints, and we purely compare heuristic performance, where the advantage of $\Gamma$-robust theory over the traditional flavor-based allocation will be discussed further below.

\begin{figure*}
\centering
\includegraphics[width=0.85\textwidth]{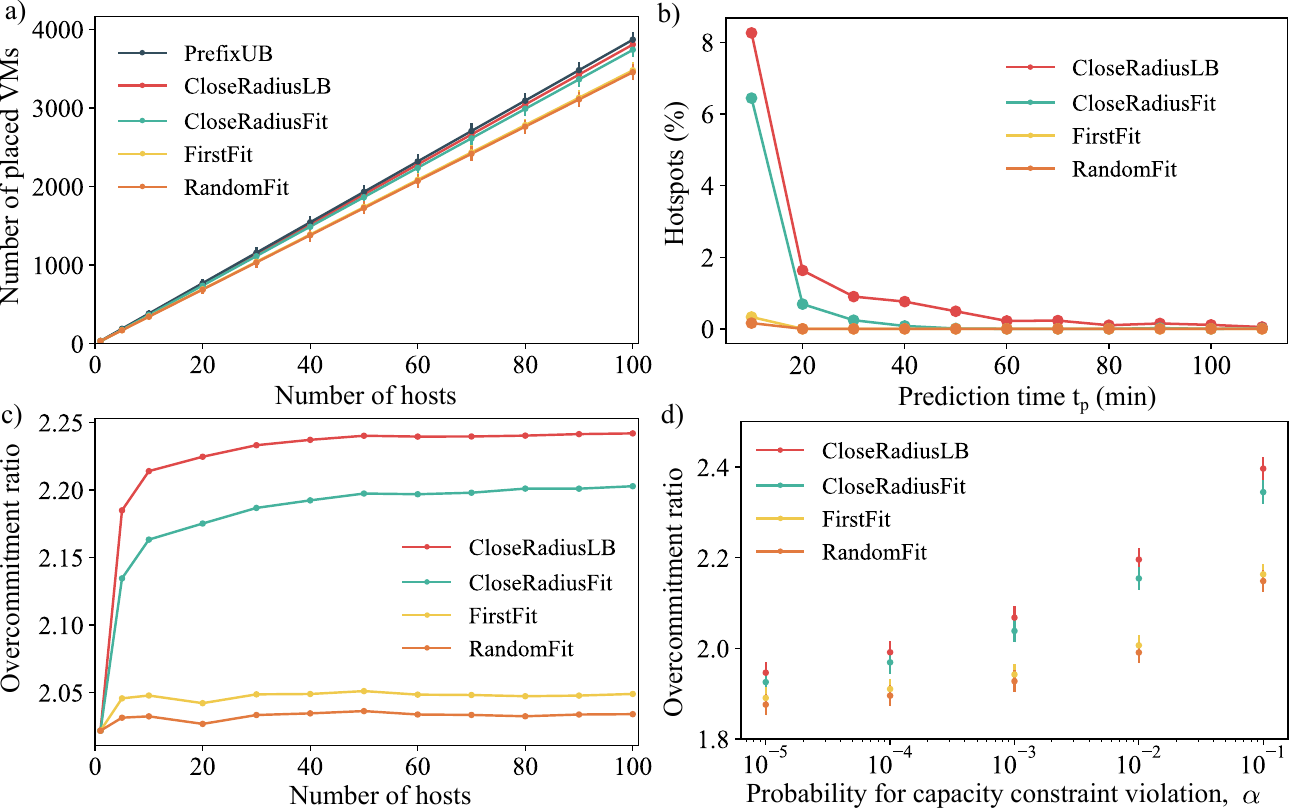}
\caption{Experiments for Hot start scenario with 100 experiments per point and $\alpha=0.05$ a) Dependence of mean number of placed VMs on the number of hosts for different heuristics and bounds. b) Dependence of the mean number of hotspots on prediction time $t_p$ for 50 hosts. c) Dependence of the mean overcommitment ratio on number of hosts. d) Dependence of the mean overcommitment ratio on the probability of capacity constraint violation $\alpha$ for 50 hosts.}
\label{fig:HotStart_results}
\end{figure*}

To evaluate the quality of our approach in addressing the under-utilization problem, we employ the metric of overcommitment ratio for a given host $h$:
\begin{equation}
    OR_h=\frac{\sum_{v \in V_h }f_v}{C},
\label{eq:OR}
\end{equation}
where $f_v$ is amount of cores in flavor of VM $v$ on host $h$. Experimental findings indicate an average OR of $2.0$, which allows to improve efficiency of space utilization 2 times under $\Gamma$-robust theory compared to flavor allocation approaches (see Figure \ref{fig:HotStart_results}c). Specifically, the CloseRadiusLB achieves an overcommitment ratio of 2.25, corresponding to a 125\% increase in mean resource utilization compared to flavor-based allocation. We observe variations in the average $OR$ for configurations with 1 to 10 hosts. This variability can be attributed to limited optimization opportunities in smaller setups and deviations from the linear relationship between the number of placed VMs and host count, as illustrated in Figure \ref{fig:HotStart_results}a. Notably, the average $OR$ depends on the level of guarantees $\alpha$: a larger $\alpha$ leads to smaller guarantees, and a higher $OR$ (see Figure \ref{fig:HotStart_results}d). The hotspots dependence on the prediction time $t_p$ (see Figure \ref{fig:HotStart_results}b) is well-correlated with utilization range change (see Figure \ref{fig:Dataset_expflow}b). These findings validate our decision to choose $t_p=40$ minutes, as it effectively mitigates hotspots and accurately predicts CPU utilization ranges.

\begin{table}[ht] 
\begin{center}
\caption{Overview of algorithms performance with $t_p=40$  minutes and $\alpha=0.05$}
\label{tab:Evaluation of the upper bound and lower bound}
\resizebox{\linewidth}{!}{
\begin{tabular}{llll}
\hline
Algorithm (X) & \begin{tabular}[x]{@{}c@{}}Placed VMs\\$V_{ph} \cdot |H|$\end{tabular}  & $Gap_{X\to LB}$ (\%) & $Gap_{X\to UB} $ (\%) \\
\hline
    PrefixUB   & $38.7\cdot|H|$ &  - & -\\
    CloseRadiusLB  & $38.1\cdot|H|$ & -     &  1.5\\
    CloseRadiusFit & $37.5\cdot|H|$  & 1.6      &  3.1  \\
    FirstFit & $34.8\cdot|H|$ &  8.7    & 10.1     \\
    RandomFit & $34.5\cdot|H|$ &  9.4    & 10.8   \\
\hline
\end{tabular}}
\end{center}
\end{table}

\subsection{Evaluation of performance in Cold start scenario}
For the cases where historical data of CPU utilization of VMs is unavailable, we employ machine learning to predict the ranges of the VMs consumption. We use the ranges for two scenarios of VMs placement: 1) Cold start, where the scheduler allocates VMs based on predicted values and 2) Semi-cold start, where in addition to predicted values, the scheduler adjusts the utilization range after $t_a$  minutes of VM active lifetime (see Section \ref{sec:VM placement with untrusted predictions}). Using the \textbf{HW Cloud-2022} dataset, we assessed VM placement performance. The number of placed VMs demonstrates a linear dependence on the number of hosts  similar to the Hot start scenario. However, the average number of placed VMs per host is less versus Hot start (see Tables \ref{tab:ColdStart} and \ref{tab:Evaluation of the upper bound and lower bound}).

\begin{table}[ht] 
\begin{center}
\caption{Overview of algorithms performance  under Cold and Semi-cold start assumptions at $t_a=40$ minutes and $\alpha=0.05$}
\label{tab:ColdStart}
\resizebox{\linewidth}{!}{
\begin{tabular}{lll}
\hline
Algorithm (X) &  \begin{tabular}[x]{@{}c@{}}Placed VMs\\$V_{ph} \cdot |H|$\end{tabular} & $Gap_{X\to LB}$ (\%)\\
\hline
    CloseRadiusFit(Semi-cold)  & $33.5\cdot|H|$ & 12.0 \\  
    FirstFit(Semi-cold) & $32.2\cdot|H|$ &  15.5\\
    RandomFit(Semi-cold) & $30.7\cdot|H|$ &  19.5  \\  
    CloseRadiusFit(Cold)  & $30.3\cdot|H|$ & 20.4 \\
    FirstFit(Cold) & $28.3\cdot|H|$ & 25.7\\
    RandomFit(Cold) & $28.1\cdot|H|$ & 26.3\\
\hline
\end{tabular}}
\end{center}
\end{table}

In the Cold start scenario, all heuristics showed approximately 20\% Gap to the Hot start solution as a result of the conservative prediction model (at fixed $t_a=40$  minutes and $\alpha=0.05$).  In the Semi-cold scenario, the Gap is improved to around 10\% due to the more accurate estimation of VMs' utilization range through the adjustment that takes place after $t_a$ minutes. 
The proportional changes of the result of all the algorithms from Cold start to Semi-cold start prove the robustness of our proposed approach. Among heuristics, CloseRadiusFit achieved the best performance towards the solution of online \problem problem, exhibiting a 12.0\% gap to the lower bound in Semi-cold start scenario, and just 1.6\% in the Hot start scenario. Note that the performance of the algorithm in Semi-cold start depends on the hyper-parameter of active time $t_a$.

To further evaluate the performance of the algorithms, we analyzed the average overcommitment ratio (OR) and the number of hotspots as functions of $t_a$ (see Figure \ref{fig:ColdStartOR_results}). In the range of $t_a \in [10, 35]$ minutes, we observed a dense packing due to the rapid update of initially conservative predictions of utilization ranges. However, since the update is based on a small number of data points, it tends to underestimate the true utilization range, resulting in the occurrence of hotspots (see Figure \ref{fig:ColdStartOR_results}b). Therefore, heuristics exhibited larger OR at the cost of larger hotspots in the range of $t_a \in [10,35]$ minutes (see Figure \ref{fig:ColdStartOR_results}). Moreover, we observed that in this range, our CloseRadiusFit algorithm does not significantly dominate the benchmark FirstFit algorithm, as the volatile utilization ranges prevent effective radius-based group packing. 
However, both CloseRadiusFit and FirstFit outperform RandomFit by employing a radius-based grouping. At each new timestamp, the utilization ranges of VMs are updated and the host loads are recalculated. The FirstFit and CloseRadiusFit heuristics are placing small VMs in the first hosts with less space. However, larger VMs cannot fit into the available space in the beginning and are consequently placed in the end of the host list. This results in a grouping pattern where VMs with small radiuses are concentrated at the beginning of the host list, whereas VMs with large radiuses are positioned at the end.
For the intermediate values of $t_a \in [35,60]$  minutes, the CloseRadiusFit exhibited obviously superior performance over the benchmark heuristics, since it takes advantage of well-estimated values of utilization range and placement by radius values. For the values of $t_a$ over 60 minutes, the algorithms converged to the pure Cold start result, where the solution is overly conservative. However, the CloseRadiusFit still outperforms the benchmark heuristics due to the meaningful prediction of the utilization ranges that allows close radius packing.

\begin{figure}
\centering
  \includegraphics[width=0.9\columnwidth]{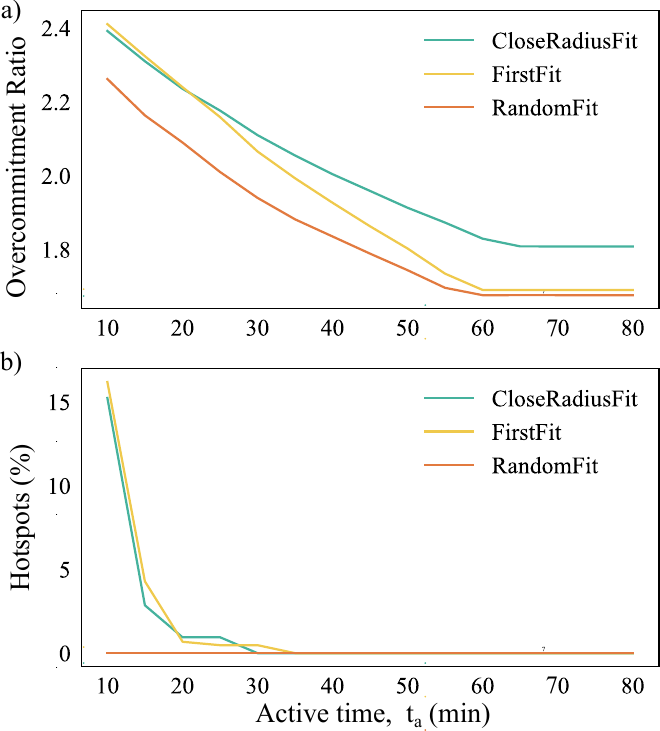}
\caption{Experiments for Semi-cold start scenario, where $\alpha=0.05$. a) Dependence of average overcommitment ratio on the active time $t_a$. b) Dependence of hotspots on active time $t_a$.}
\label{fig:ColdStartOR_results}
\end{figure}

\section{Conclusion} \label{sec:conclusion}
In this work, we define and formulate the \problem problem for VMs placement with dynamic CPU utilization, where the $\Gamma$-robustness theory is applied to manage the uncertainty. Our objective is to maximize the number of provisioned VMs into a fixed number of hosts, while maintaining a certain probability of not violating physical capacity constraints. In offline settings, we make key contributions through a novel MILP model and two high-quality bounds, i.e., PrefixUB and CloseRadiusLB. These methods lay the groundwork for online heuristic design. In online settings, we investigate three regimes of operations: Hot start, Cold start and Semi-cold start. Comprehensive experiments on Huawei's real trace demonstrate that our CloseRadiusFit performs well in all the scenarios and significantly advances the classical online algorithms, with gaps to upper bound on the optimal solution of 3.1\% in Hot start, 13.5\% in Semi-cold start, and 21.5\% in Cold start. These gains directly enable high VM packing density in real data centers and benefit the cloud service providers.

Regarding future research, we intend to examine the impact of other resource dimensions, such as memory, disk, and network. This will transform the problem into a variant of the probabilistic multi-dimensional bin packing problem, and by taking into account the interrelated influence of these different dimensions, the problem becomes more complex. Besides, more advanced predictive models tailored to this problem could further enhance the performance of Cold start and Semi-cold start scenarios. Additionally, the applications beyond VM placement, such as container orchestration, present additional domains to apply and evaluate our techniques.

\section{Data availability} \label{sec:dataset}
The datasets used in this article are available by the link \url{https://github.com/huaweicloud/HUAWEICloudPublicDataset}. The code with algorithms and bounds is available by the link \url{https://github.com/andreigudkov/CloseRadiusFit}.

\bibliographystyle{IEEEtranN}
\bibliography{refs}

\begin{IEEEbiography}
[{\includegraphics[width=1in,height=1.25in,clip,keepaspectratio]{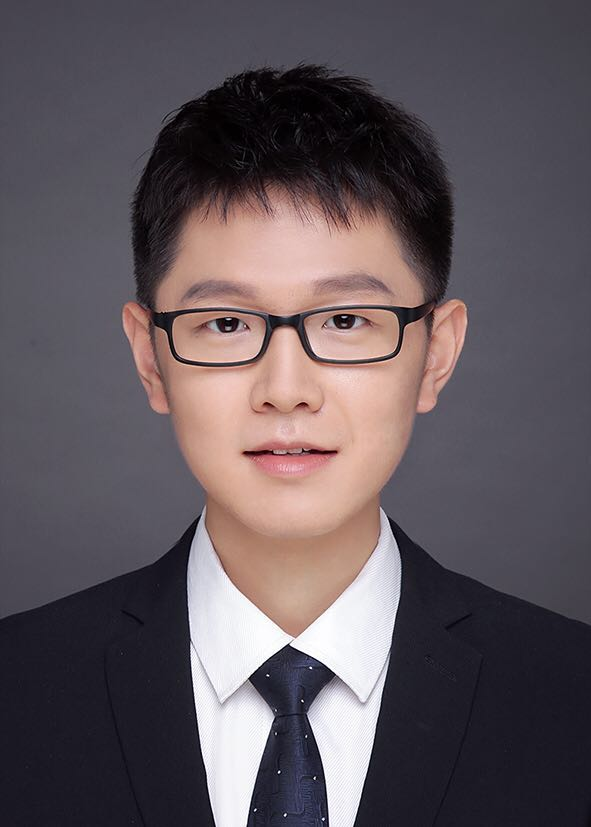}}]{Jiaxi Wu}
received his B.S. degree in Management Science from Tianjin University, Tianjin, China, in 2016, and the Ph.D degree in Industrial Engineering and Management from Peking University, Beijing, China, in 2021. He then worked as a postdoctoral researcher at Huawei Cloud, where he gained valuable industry experience in cloud computing optimization. He is currently a data scientist and architect at China Mobile. His research is devoted to the theory and development of large-scale optimization algorithms, robust optimization and data science, with applications in the fields of intelligent manufacturing, cloud computing and big data systems. 
\end{IEEEbiography}
\vskip -1\baselineskip plus -1fil
\begin{IEEEbiography}[{\includegraphics[width=1in,height=1.25in,clip,keepaspectratio]{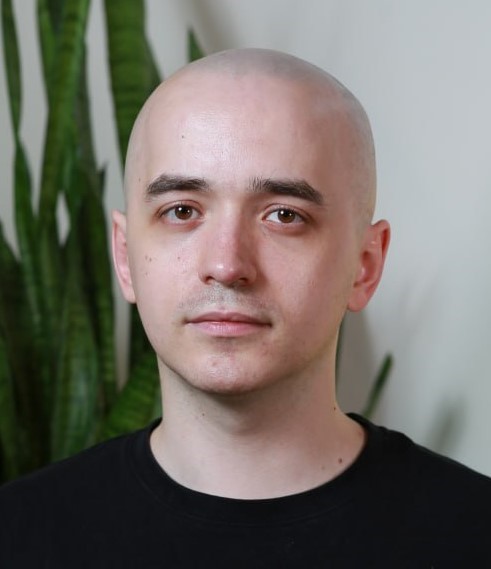}}]{Pavel Popov}
received his M.S. degree in Mathematics from the National Research University Higher School of Economics, Moscow, in 2015. He has also completed the Yandex School of Data Analysis in 2022. Previously, he conducted research in Algebraic Geometry, focusing on the study of cubic hypersurfaces. Currently, he is working in the field of Combinatorial Optimization, specifically on different Cloud Scheduling Problems. His research aims to develop efficient algorithms for resource allocation in cloud computing systems.
\end{IEEEbiography}
\vskip -1\baselineskip plus -1fil
\begin{IEEEbiography}
[{\includegraphics[width=1in,height=1.25in,clip,keepaspectratio]{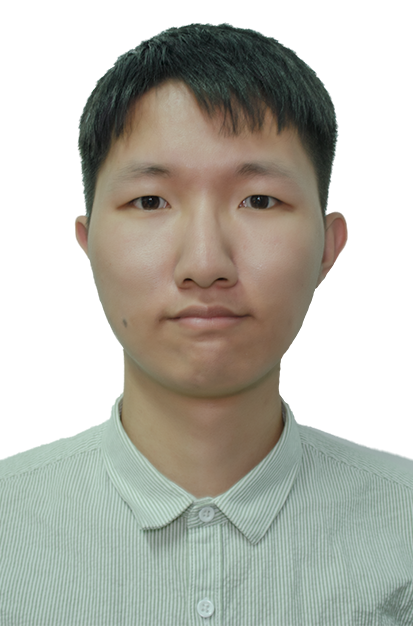}}]{Wenquan Yang}
received his B.S. degree in Computer Science and Technology from Northeastern University, Shenyang, China, in 2021. He is currently an algorithm engineer in Huawei Cloud Computing. His research interest includes the theory of machine learning and deep learning, with applications in user portrait and resource allocation. 
\end{IEEEbiography}
\vskip -1\baselineskip plus -1fil
\begin{IEEEbiography}[{\includegraphics[width=1in,height=1.25in,clip,keepaspectratio]{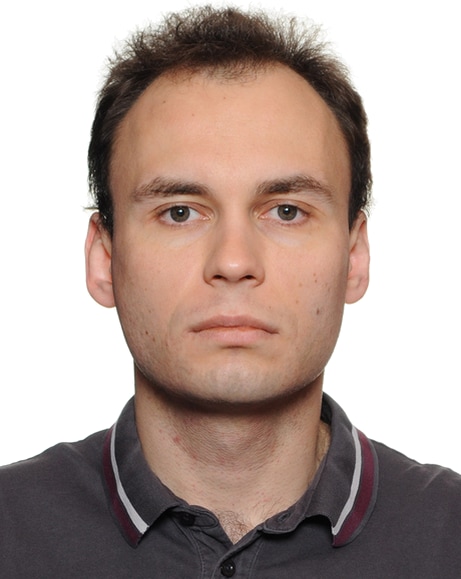}}]{Andrei Gudkov}
received his M.Sc. degree from the faculty of Computational Mathematics and Cybernetics of Moscow State University, Russia, in 2009. He held multiple engineering positions in hi-tech companies specializing in the areas of full-text search engines, distributed computing and big data. Currently he is employed at Mathematical Modeling Lab of Huawei Moscow Research Center, focusing on resource usage optimization in Huawei Cloud. Possessing deep experience in computing, he models cloud environment at all levels and designs practical scheduling algorithms.
\end{IEEEbiography}
\vskip -1\baselineskip plus -1fil

\begin{IEEEbiography}[{\includegraphics[width=1in,height=1.25in,clip,keepaspectratio]{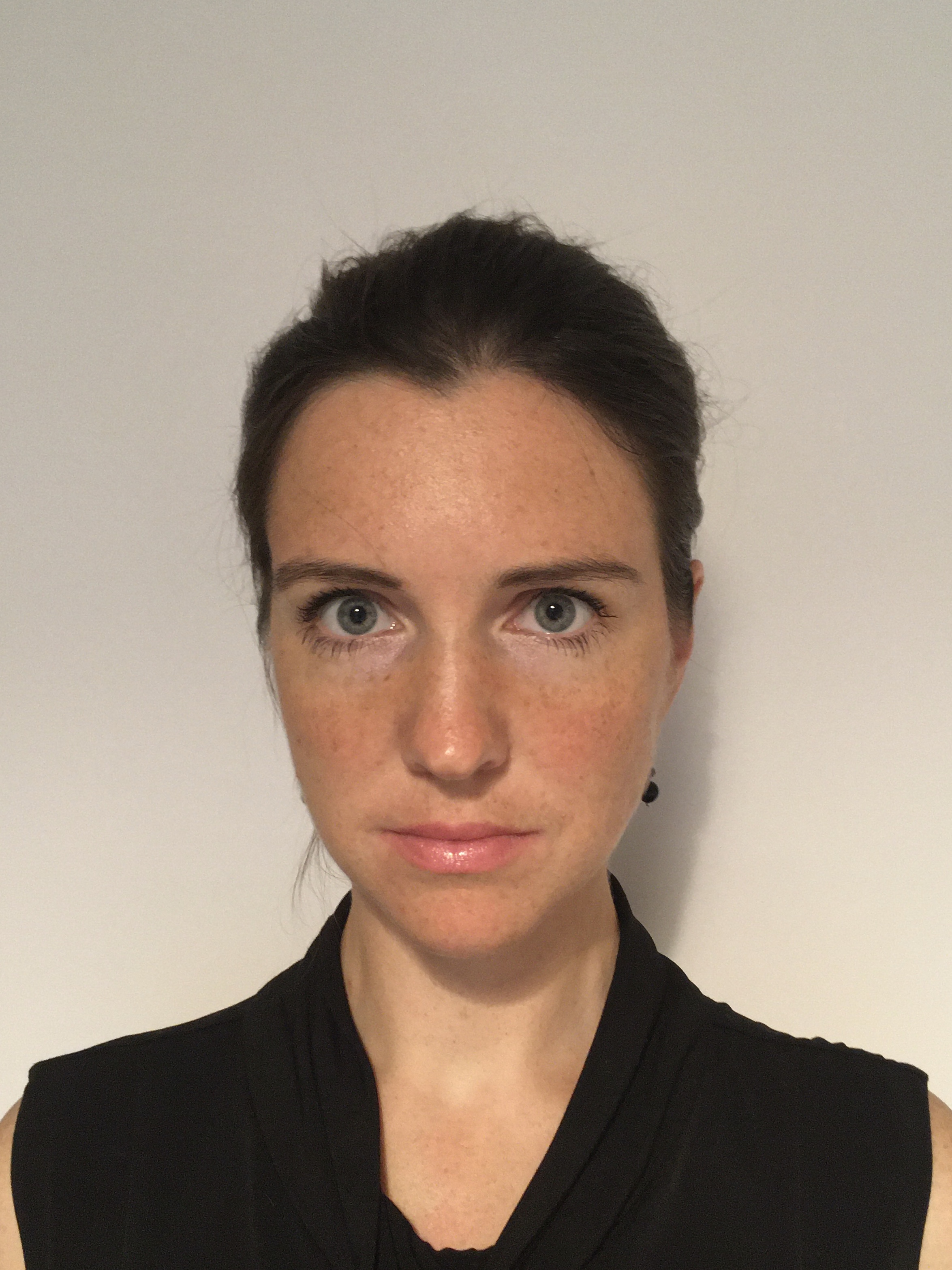}}]{Elizaveta Ponomareva}
received her Ph.D. in Mathematics from Moscow State University in 2016. She also graduated from Yandex School of Data Analysis. During her time at the university, she conducted research in Algebraic Geometry and Invariant Theory. Afterward, Dr. Ponomareva worked as a data analyst at a high-tech company. Currently, her research is focused on resource allocation in cloud computing systems.

\end{IEEEbiography}
\vskip -1\baselineskip plus -1fil
\begin{IEEEbiography}
[{\includegraphics[width=1in,height=1.25in,clip,keepaspectratio]{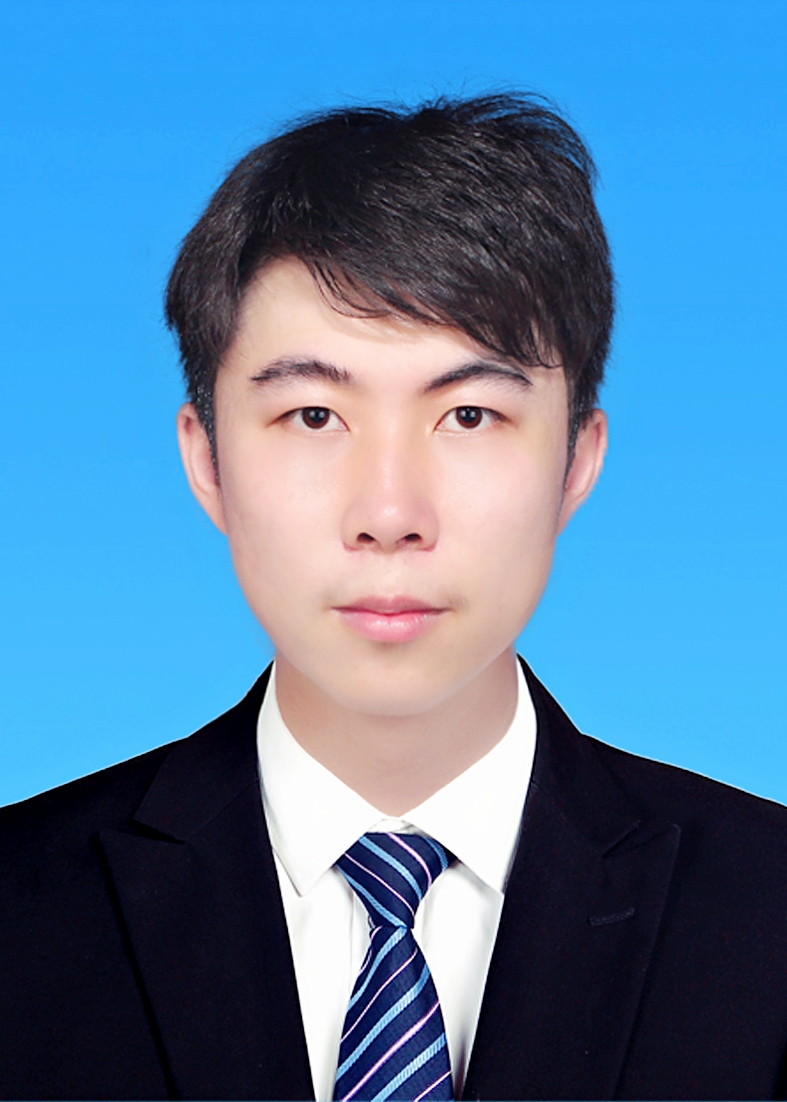}}]{Xinming Han}
received the B.S. degrees in industrial engineering from Southwest Jiaotong University, Chengdu, China, in 2020. Currently he is a Ph.D. candidate at the Department of Industrial Engineering and Management, Peking University, Beijing, China. His research interests focus on cloud resource scheduling and optimization.
\end{IEEEbiography}

\vskip -1\baselineskip plus -1fil
\begin{IEEEbiography}
[{\includegraphics[width=1in,height=1.25in,clip,keepaspectratio]{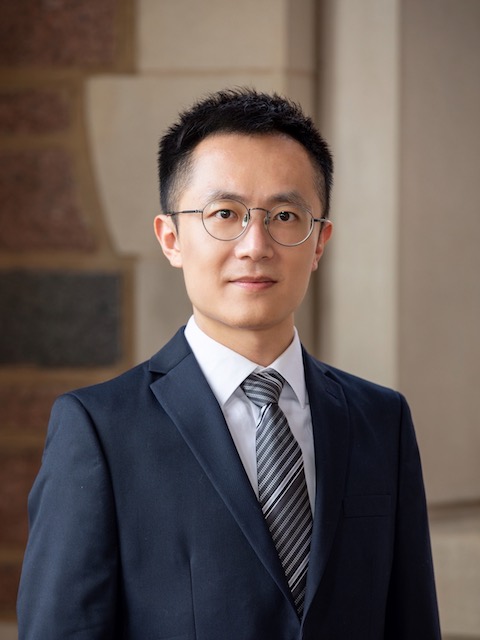}}]{Yunzhe Qiu}
Dr. Yunzhe Qiu is currently an assistant professor in the Department of Information Management at Peking University. He received a Ph.D. degree in Supply Chain and Operations Technology from Olin Business School, Washington University in St. Louis in 2022, a Master degree in Industrial Engineering and Management from Peking University in 2016, and a B.S. degree in Industrial Engineering from Tsinghua University in 2013. His research interests include Interface of Finance, Operations and Risk Management (iFORM), supply chain management, data-driven operations management, dynamic program, and simulation optimization. He won the “Best Practice” at the 14th International Annual Conference of the CSAMSE in 2022, the Honorable Mention Award by IISE in 2020, and the Best Paper Award from RAS at IEEE CASE in 2014. 
\end{IEEEbiography}
\vskip -1\baselineskip plus -1fil
\begin{IEEEbiography}
[{\includegraphics[width=1in,height=1.25in,clip,keepaspectratio]{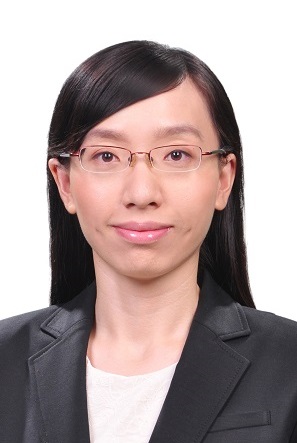}}]{Jie Song}
received the B.S. degree in applied mathematics from Peking University, Beijing, China, in 2004, and the M.S. and Ph.D. degree in industrial engineering from Tsinghua University, Beijing, in 2007 and 2010, respectively. She is an Professor with the Department of Industrial Engineering and Management, Peking University. Her research interests are simulation optimization, stochastic modeling in the application areas of logistics, healthcare and production and cloud computing.
\end{IEEEbiography}

\vskip -1\baselineskip plus -1fil
\begin{IEEEbiography}[{\includegraphics[width=1in,height=1.25in,clip,keepaspectratio]{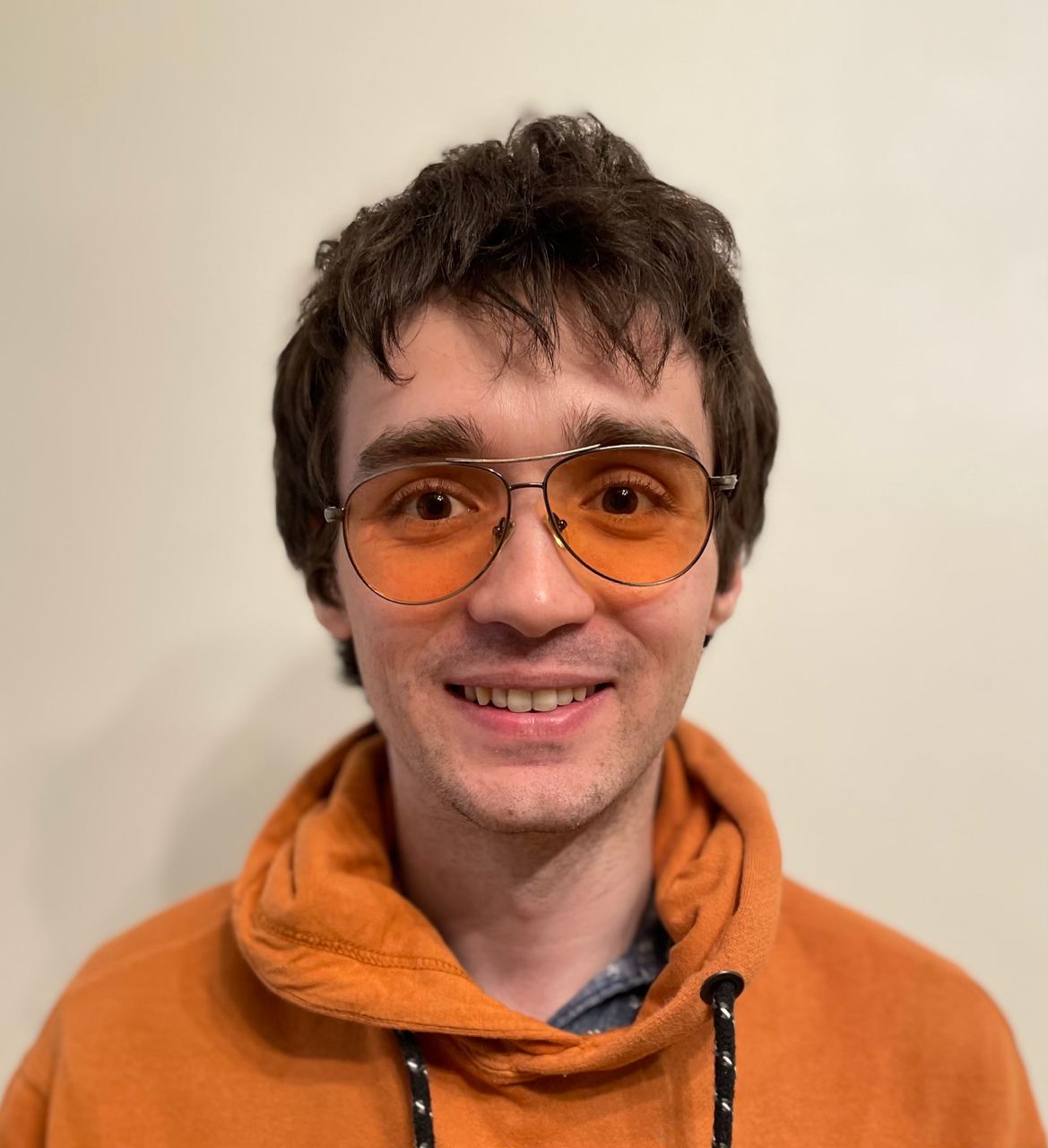}}]{Stepan Romanov}
holds both a B.S. and an M.S. in computer science from the Moscow Institute of Physics and Technologies. He further pursued his academic passions and obtained an MS and a Ph.D. in physics from the Skolkovo Institute of Science and Technologies. Dr. Romanov is a multidisciplinary scholar with a keen interest in diverse areas such as computational chemistry, experimental physics, robotics, and computer science. Currently, he is actively engaged in researching and tackling challenges in cloud computing systems.
\end{IEEEbiography}
\clearpage
\appendix
\input{AppendixClean}

\end{document}

%% file: AppendixClean.tex
\title{Appendix for article: Hotspot-Aware Scheduling of Virtual Machines with Overcommitment for Ultimate Utilization in Cloud Datacenters}

\author{Jiaxi Wu, Pavel Popov, Wenquan Yang, Andrei Gudkov, Elizaveta Ponomareva, Xinming Han, Yunzhe Qiu, Jie Song, Stepan Romanov 
\IEEEcompsocitemizethanks{\IEEEcompsocthanksitem Jiaxi Wu, Pavel Popov, Wenquan Yang, and Stepan Romanov are from Huawei Technologies Company Ltd  \protect\\
E-mail: stepan.romanov@skolkovotech.ru \protect\\
Xinming Han, Yunzhe Qiu, Jie Song are from Peking university.
}
}

\maketitle

\IEEEdisplaynontitleabstractindextext
\IEEEpeerreviewmaketitle
\appendices

\section{Proof of Theorem \ref{thm:minimal_s_fordominance}}
\label{thm:proof_minimal_s_fordominance}
To prove the theorem we rely on two facts:
 \begin{enumerate}
          \item For any symmetric $\hat{v} \succeq u $ with center $\hat{v}c=uc+s_v/2$ it holds that $F_u(x) \ge F_{sym(u)}(x-s_v)$ \item For any $s_v$ \textit{s.t.} $F_u(x) \ge F_{sym(u)}(x-s_v)$ there exists a symmetric $\hat{v} \succeq u$ with center $\hat{v}c= uc + s_v/2$  
\end{enumerate}
From these two facts and requirement on minimal value of $s$ \textit{s.t.} $F_u(x) \ge F_{sym(u)}(x-s) $ the theorem is immediately proven.   
To prove fact 1 we start with Def. \ref{def:stochastic_dominant} of stochastic domination of $\hat{v}$ 
\begin{equation}
    F_{\hat{v}}(x) \le F_u(x).
\label{eq:dominance}
\end{equation}
Applying reflection of both sides around point 
$(2\cdot \hat{v}c, 0.5)$ we have 
\begin{equation}
    F_{\hat{v}}(x) \ge 1-F_{u}(2\cdot \hat{v}c - x).
\label{eq:sym_fact1_2}
\end{equation}
Noting that $\hat{v}c = uc + s_v/2$ and inserting it in Eq. \ref{eq:sym_fact1_2} we obtain
\begin{equation}
    F_{\hat{v}}(x) \ge 1-F_{u}(2\cdot uc + s_v - x) = F_{sym(u)}(x-s_v).
\label{eq:fact1_int}
\end{equation}
Combining Eqs. \ref{eq:dominance} and \ref{eq:fact1_int} we obtain final inequality that we tend to prove
\begin{equation}
    F_u(x) \ge F_{sym(u)}(x-s_v).
\end{equation}

To prove fact 2 we just introduce such symmetric $\hat{v}$ with center $uc + s_v/2$
\begin{equation}
F_{v}(x) = 
\begin{cases}
    F_{sym(u)}(x-s_v), x < uc + s_v/2 \\
    F_{u}(x), x \ge uc + s_v/2 
\end{cases}
\end{equation}

\section{Proof of Theorem \ref{thm:overhead2}}
\label{appendix: proof of overhead}
Given a queue $Q$, consider VMs in $Q$ sorted by radiuses in descending order.  Denote by $P_i$ the prefix of length $i$. 
Let us call the $MaxSet(\mu)$ the union over all hosts $MaxSet_h(\mu)$ in mapping $\mu$.
Our goal is to find the lower bound $ur_{lb}(P_i)$ on probabilistic overhead, that estimates  $MaxSet(\mu)$ in any mapping $\mu$
\begin{equation}
\label{eq:maxset_goal}
    \sum\limits_{v \in MaxSet(\mu)} ur_v \ge ur_{lb}(P_i) \quad \forall \mu:P_i \rightarrow H
\end{equation}

We are aimed to propose iterative procedure towards calculation of $ur_{lb}(P_i)$. 
Starting with $ur_{lb}(V_0)=0$ we compute all $ur_{lb}(P_i)$ values using forward induction.
Suppose that $ur_{lb}(P_i)$ has already been computed.
There are two cases for changing of $\Gamma^{lb}$ and calculations for $ur_{lb}$ after adding a new VM: 
\begin{enumerate}
  \item If $\Gamma^{lb}(P_{i+1})\!=\!\Gamma^{lb}(P_{i})$ \\
        then $ur_{lb}(P_{i+1})= ur_{lb}(P_{i})$ 
  \item If $\Gamma^{lb}(P_{i+1})\!\geq\!\Gamma^{lb}(P_{i})\!+\!1$\\
  then $ur_{lb}(P_{i+1})\!=\!ur_{lb}(P_{i})\!+\!ur_{i+1}$. 
\end{enumerate}

To prove that iterative procedure satisfies Eq. \ref{eq:maxset_goal} let us introduce a few definitions that we will use throughout the proof. 
\begin{definition}
\label{app:def_exeeds}
Let $V$ and $W$ be two sequences of VMs sorted by radiuses in non-increasing order. Denote by $vr_i, wr_i$ radiuses of the corresponding VMs in the sets. We say that $V$ \emph{dominates} $W$  ($V \succeq W$) if
\begin{enumerate}
    \item $|V| \geq |W|$
    \item $vr_i \geq wr_i, \forall i \leqslant |W|$
\end{enumerate}
\end{definition}

\begin{properties} If $V \succeq W$ the following obvious properties holds:
\label{app:prop}
\begin{enumerate}
    \item $\sum\limits_{v \in V} ur_v \ge \sum\limits_{w \in W} ur_w$
    \item $V \cup v \succeq W \quad \forall v$ 
    \item If $ur_{v} \ge ur_w$, then $V \cup v \succeq W \cup w$
\end{enumerate}
\end{properties}

To satisfy Eq. \ref{eq:maxset_goal} we search for $S$ that is dominated by $MaxSet(\mu)$, \textit{i.e.}  $\sum\nolimits_{v \in MaxSet(\mu)} ur_v \ge \sum\nolimits_{v \in S} ur_v $. Our iterative procedure for calculation of $ur_{lb}(P_n)$ convert to procedure of $S$ construction: 

\begin{enumerate}
  \item $\Gamma^{lb}(P_{0}) = 0$. Set $S(P_0) = \emptyset$
  \item If $\Gamma^{lb}(P_{i+1})\!=\!\Gamma^{lb}(P_{i})$.
  Set $S(P_{i+1})= S(P_{i})$. 
  
  \item If $\Gamma^{lb}(P_{i+1})\!\geq\!\Gamma^{lb}(P_{i})\!+\!1$. 
  Set $S(P_{i+1}) = S(P_{i}) \cup v_{i+1}$. 

\end{enumerate}

The sum of radiuses of $S(P_{i})$ is equal to $r_{lb}(P_i)$. The number of elements $|S(P_{i})|$ is not greater than the lower bound $\Gamma^{lb}(P_{i})$. To finish the proof of the theorem, it is remained to prove: 

\begin{lemma}
$MaxSet(\mu) \succeq S(P_i)$ for any $P_i$ and mapping $\mu$.
\end{lemma}
\begin{proof}
The proof is by induction. The base is obvious. Suppose that $S(P_{i})$ is dominated by $MaxSet(\mu)$. 
Consider any feasible assignment $\mu\colon P_{i+1}\rightarrow H$ and corresponding reduced assignment $\mu'\colon P_{i} \rightarrow H$. $MaxSet(\mu)$ is either equal to $MaxSet(\mu')$ or is extended by one VM $v$. 
It implies that $MaxSet(\mu) \succeq MaxSet(\mu')$ because of Property 1.2.
We need to show that $MaxSet(\mu) \succeq S(P_{i+1})$.
Two transition cases are needed  to be considered. 

\subsubsection*{Case 1 \texorpdfstring{$\Gamma^{lb}(P_{i+1})\!=\!\Gamma^{lb}(P_{i})$.}{Lg}} 
It implies that $S(P_{i+1}) = S(P_{i})$. 
Combining  $MaxSet(\mu) \succeq MaxSet(\mu')$ with the induction step $MaxSet(\mu')\succeq S(P_{i})$, we finally prove that  $MaxSet(\mu) \succeq S(P_{i})=S(P_{i+1})$. 

\subsubsection*{Case 2 \texorpdfstring{$\Gamma^{lb}(P_{i+1})\!\geq\!\Gamma^{lb}(P_{i})\!+\!1$.}{Lg}}
We have  $S(P_{i+1}) = S(P_{i}) \cup v_{i+1}$ and need to consider two cases.
\begin{enumerate}
    \item 
    \begin{equation}
    \label{eq:app:theorem2_init}
        MaxSet(\mu)= MaxSet(\mu') \cup w,
    \end{equation}
    where $w$ is either a new arrived VM $v_{i+1}$ or a previously placed VM from $MinSet(\mu')$. Using the induction step $MaxSet(\mu')\succeq S(P_{i})$ and Property \ref{app:prop}.3 we have 
    \begin{equation}
    \label{eq:app:theorem2_1}
        MaxSet(\mu') \cup w \succeq S(P_{i}) \cup v_{i+1} ,
    \end{equation}
    because VMs are sorted ($w \ge v_{i+1} $). Using 3rd property of $S$ construction $S(P_{i+1}) = S(P_{i}) \cup v_{i+1}$ we substitute it to Eq. \ref{eq:app:theorem2_1} and have $MaxSet(\mu') \cup w \succeq S(P_{i+1})$, and using Eq. \ref{eq:app:theorem2_init} we finally prove that $MaxSet(\mu) \succeq S(P_{i+1})$.

    \item
    \begin{equation}
    \label{eq:app:theorem2c2_n1}
       MaxSet(\mu) = MaxSet(\mu')
    \end{equation}
    We decompose $ MaxSet(\mu')$ as $X \cup v$, where $v$ is the smallest VM from $MaxSet(\mu')$, combining it with induction step $MaxSet(\mu')\succeq S(P_{i})$ we have
    \begin{equation}
    \label{eq:app:theorem2c2_n2}
       MaxSet(\mu) = MaxSet(\mu') = X \cup v  \succeq S(P_{i}).
    \end{equation}   
    We want using Definition \ref{app:def_exeeds} and Eq. \ref{eq:app:theorem2c2_n2} prove that 
    \begin{equation}
    \label{eq:app:theorem2c2_n3}
       X \succeq S(P_{i}).
    \end{equation}   
    From properties of $\Gamma^{lb}$ and Eq. \ref{eq:app:theorem2c2_n1} we have, $\Gamma^{lb}(P_i) \le |S(P_i)|$ and $|MaxSet(\mu')| \ge \Gamma^{lb}(P_{i+1}) $. Using 
    condition of Case 2 ($\Gamma^{lb}(P_{i+1})\!\geq\!\Gamma^{lb}(P_{i})\!+\!1$) we have
    \begin{equation}
        |MaxSet(\mu')| \ge \Gamma^{lb}(P_{i+1}) \geq\!\Gamma^{lb}(P_{i})\!+\!1 \ge |S(P_i)| + 1
    \end{equation}   
    Therefore $S(P_i)$ has less elements then $MaxSet(\mu')$ and as far as $v$ is the smallest VM from $MaxSet(\mu')$ the  Eq. \ref{eq:app:theorem2c2_n3} is correct 
    because of Definition \ref{app:def_exeeds} and Eq. \ref{eq:app:theorem2c2_n2}.

    Adding $v$ and $v_{i+1}$ to Eq. \ref{eq:app:theorem2c2_n3} and using Property \ref{app:prop}.3 we have 
    \begin{equation}
    \label{eq:app:theorem2c2_n4}
       X \cup v \succeq S(P_{i}) \cup v_{i+1},
    \end{equation}     
    that is valid because VMs are sorted in descending order and $rv \ge rv_{i+1}$. Finally substituting Eq. \ref{eq:app:theorem2c2_n2} and definition of $S(P_{i+1}) = S(P_{i}) \cup v_{i+1}$ into
    Eq. \ref{eq:app:theorem2c2_n4} we have  
    \begin{equation}
        MaxSet(\mu) =  X \cup v  \succeq   S(P_{i}) \cup v_{i+1} =S(P_{i+1}).
    \end{equation}   
    Therefore we proved that $MaxSet(\mu) \succeq  S(P_{i+1})$

\end{enumerate}

\section{Computation of \texorpdfstring{$\mathbf{\Gamma}$}{Lg} and \texorpdfstring{$\mathbf{\widetilde{\Gamma}}$}{Lg}} 
\label{appendix: computation of Gamma and widetilde Gamma}

The number of VMs on host $N$, parameter $\Gamma$ and acceptable probability of constraint violation $\alpha$ are connected with the following expression $B(N,\Gamma) = \alpha$ (\cite{BERTSIMAS}), where:

\begin{equation}
\label{eq:B_appendix}
    B(N,\Gamma) = \frac{1}{2^N}\cdot \biggl\{ (1-\mu) {N\choose \left \lfloor \nu \right \rfloor} + \sum\limits_{l=\left \lfloor \nu \right \rfloor + 1)}^N {N\choose l} \biggr\}
\end{equation}
and
\begin{equation*}
    \mu = \nu - \left \lfloor \nu \right \rfloor, \;\; \nu = (\Gamma + N)/2.
\end{equation*}
In the work of Bertsimas \textit{et. al.} parameter $\Gamma$ takes real values, however in our scenario it is a natural number representing the number of VMs in $MaxSet_h$. To obtain $\Gamma(N)$ we numerically solved $B(N,\Gamma) = \alpha$ and rounded up the obtain values (see Figure \ref{fig:GammaExample}a). Also, for calculation of PrefixUB we required a concave approximation of the $\Gamma(N)$, which we obtained using linear programming technique described in Sec. \ref{sec:concave}.

\end{proof}

\section{Proof of Lemma \ref{lemma:radiuses seperation}}
\label{appendix: proof of radiuses seperation}

We aim to split $2 N$ VMs between two hosts and minimize probabilistic overhead. Lets sort VMs by radiuses $ur_1 \geq ... \geq ur_{2N}$.
First $\Gamma$ VMs always belong to $MaxSet=Maxset_1 \cup MaxSet_2$ in any mapping, whereas VMs from $N+\Gamma+1$ to $2N$ belong to $MinSet=MinSet_1 \cup Minset_2$ in any mapping (see Figure \ref{fig:proof_lemma1}). 
To minimize the probabilistic overhead, among the rest VMs from $\Gamma+1$ to $N+\Gamma$ the smallest VMs from $N+1$ to $N + \Gamma$ have to be in $MaxSet$ and the rest in $MinSet$. Therefore we constructed optimal $MinSet$ and $MaxSet$ and packing of the hosts using VMs order as stated in Lemma \ref{lemma:radiuses seperation} is achievable and optimal. 

\begin{figure}[ht]
  \centering
  \includegraphics[width=0.70\columnwidth]{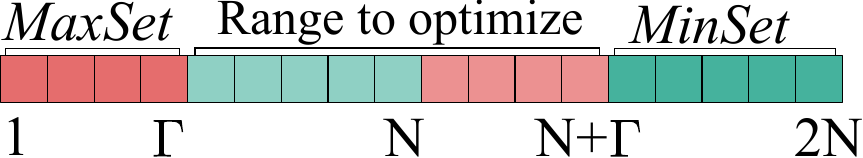}
  \caption{Illustration for the proof of Lemma \ref{lemma:radiuses seperation}. }
  \label{fig:proof_lemma1}
\end{figure}

\section{Proof of Lemma \ref{lemma:best radiuses values}}
\label{appendix: proof of best radiuses values}

\begin{proof}
    Since VMs from $MinSet$ are uninvolved in the calculation of probabilistic load of the host, all their radiuses must be equal to
    $ur_{\Gamma}$ in optimal solution; because otherwise $ur_{MinSet}$ is improvable by increase of radius from MinSet to $ur_{\Gamma}$. 
    Therefore our goal is to maximize $ur_{\Gamma}$, which is clearly achieved when
    all $ur_i$ in MaxSet are the same; because otherwise we can increase $ur_{\Gamma}$ at the cost of decrease of another VM from MaxSet by keeping capacity equal to C.
    In other words, the optimal solution is achieved when all VMs from $MaxSet$ and $MinSet$ are equal to $ur_{opt}=(C - \sum\nolimits_{v=1}^{N}{uc_v})/\Gamma_N$. 
\end{proof}